\documentclass[sigconf]{acmart}

\usepackage{pdflscape}
\usepackage{ifxetex,ifluatex}
\usepackage{fixltx2e} 

\usepackage{hyperref}
\usepackage{graphicx}

\usepackage{tikz}
\usepackage{semantic}
\usepackage{bbm}
\usepackage[english]{babel}
\usepackage{mathtools}
\usepackage[]{agda}

\usepackage{tcolorbox}

\usepackage{todonotes, comment}

\usepackage{mathpartir}

\usepackage{titlesec}

\newtheorem{example}{Example}
\newtheorem{theorem}{Theorem}
\newtheorem{remark}{Remark on Notation}

\newtheorem{case}{Inductive Case}
\newtheorem{base}{Base Case}
\newtheorem{definition}{Definition}

\usepackage{multicol}

\DeclareUnicodeCharacter{393}{$\Gamma$}
\DeclareUnicodeCharacter{2081}{$^1$}
\DeclareUnicodeCharacter{03BB}{$\lambda$}
\DeclareUnicodeCharacter{2200}{$\forall$}

\newcommand{\tnorm}[2]{#2 \downarrow_{#1}}
\newcommand{\pnorm}[1]{\tnorm{\mplus}{#1}}

\newcommand{\denote}[1]{\llbracket #1 \rrbracket}

\newcommand{\st}{<:}

\newcommand{\strans}{\leadsto}
\mathchardef\mplus="2B

\newcommand{\highlightg}[1]{%
\colorbox{green!30}{$\displaystyle#1$}}

\copyrightyear{2020}
\acmYear{2020}
\setcopyright{acmlicensed}\acmConference[PPDP '20]{22nd International Symposium on Principles and Practice of Declarative Programming}{September 8--10, 2020}{Bologna, Italy}
\acmBooktitle{22nd International Symposium on Principles and Practice of Declarative Programming (PPDP '20), September 8--10, 2020, Bologna, Italy}
\acmPrice{15.00}
\acmDOI{10.1145/3414080.3414094}
\acmISBN{978-1-4503-8821-4/20/09}

\def\minus{-}

\AtBeginDocument{%
  \providecommand\BibTeX{{%
    \normalfont B\kern-0.5em{\scshape i\kern-0.25em b}\kern-0.8em\TeX}}}

\begin{document}

\title{Proof-Carrying Plans: a Resource Logic for AI Planning}


\author{Alasdair Hill}
\affiliation{%
  \institution{Heriot-Watt University}
  \city{Edinburgh}
  \country{Scotland}}
\email{ath7@hw.ac.uk}

\author{Ekaterina Komendantskaya}
\affiliation{%
    \institution{Heriot-Watt University}
    \city{Edinburgh}
    \country{Scotland}
}
\email{ek19@hw.ac.uk}

\author{Ronald P. A. Petrick }
\affiliation{%
    \institution{Heriot-Watt University}
    \city{Edinburgh}
    \country{Scotland}
}
\email{R.Petrick@hw.ac.uk}


\begin{abstract}Planning languages have been used successfully in AI for several decades.
  Recent trends in AI verification and Explainable AI have raised the question of whether AI planning techniques can be verified.
  In this paper, we present a novel resource logic, the \emph{Proof Carrying Plans (PCP) logic} that can be used to verify plans produced by AI planners.
  The PCP logic takes inspiration from existing resource logics (such as Linear logic and Separation logic)
  as well as Hoare logic
  when it comes to modelling states and resource-aware plan execution.
  It also capitalises on the Curry-Howard approach to logics, in its treatment of plans as functions and plan pre- and post-conditions as types.
  This paper presents two main results. From the theoretical perspective,
  we show that the PCP logic is sound relative to the standard possible world semantics used in AI planning. From the practical perspective,
  we present a complete Agda formalisation of the PCP logic and of its soundness proof. Moreover, we showcase the Curry-Howard, or functional, value of this implementation
  by supplementing it with the library that parses AI plans into Agda's proofs automatically. We
  provide evaluation of this library and the resulting Agda functions.
\end{abstract}


\maketitle

\section{Motivation}
Planning is a research area within AI that studies automated generation of plans from
symbolic domain and problem specifications.
AI planners came into existence in the 1970s as an intersection between general problem solvers \cite{ernst1969gps},
situation calculus \cite{mccarthy1981some} and theorem proving \cite{green1969theorem}.
One of the most popular early planners was the Stanford Research Institute Problem Solver (STRIPS) \cite{strips}
which was created to address the problems faced by a robot in rearranging objects and in navigating.

In STRIPS, a planner is given a description of  an \emph{initial state} (of the ``world'')
and a \emph{goal state}.  For example, Figure~\ref{fig:pddl-blocksworld2} defines the initial state that has blocks a and b on the table,
and the goal state -- the blocks  assembled in a stack.
A solution to a planning
problem is a sequence of actions, which is simply referred to as a \emph{plan}.
For example, a solution to the planning problem of  Figure~\ref{fig:pddl-blocksworld2} is the following plan: \emph{pickup a  from the table}, then
\emph{putdown a on b}.

 \begin{figure}[t!]
    \centering
       \begin{verbatim}
    (define (problem blocksworld)
      (:domain blocksworld)
      (:objects a b)
      (:init (onTable a)
             (onTable b)
             (clear a)
             (clear b)
             (handEmpty))
      (:goal (and (on a b) (onTable b))))
    \end{verbatim}
    \caption{\emph{``BlocksWorld'' Planning Problem Description. Initial state: two blocks, a and b, are lying ``clear'' (i.e. unobstructed) on a table, and a robot hand is empty. Goal state: block a is on block b.
        }}
    \label{fig:pddl-blocksworld2}
  \end{figure}

Many versions of planning languages were proposed, and
 the \emph{Planning Domain and Definition Language (PDDL)}~\cite{mcdermott1998pddl} aimed to standardise them.
One notable design decision of PDDL is the splitting of the planning problem into
\emph{domain} and \emph{problem descriptions}. The domain description describes generally predicates and admissible actions (as shown in  Figure~\ref{fig:pddl-blocksworld}), while
the problem description defines specific initial and goal states (as shown in  Figure~\ref{fig:pddl-blocksworld2}).

PDDL has many extensions over regular STRIPS syntax
with the latest version supporting types, numerical functions, equality,
conditionals, concurrency, temporal planning and more.
Among applications are: reasoning about knowledge, belief and causality,
planning allocation of resources, modelling perception of the real world,
program synthesis and implementations of multi-agent systems \cite{hendler1990ai,wilkins2014practical}.
Many of these applications of planning are used in real-world
environments where the verification of plan correctness is essential for
successful and safe operation.

Verification and validation of AI planning languages \cite{bensalem2014verification}
is a rich field of research.
One may verify  domain models, planning algorithms, or the produced plans.

Verifying domain models \cite{penix1998using,long2009planning} seeks to validate
whether  domain descriptions accurately capture (expert) knowledge about the world.
This can be done by performing test based verification of input and
output specifications to check the domain performs as expected.
Alternatively, some approaches ensure properties
that should hold across many domains such as enforcing that the actions cannot
lead to an inconsistent state.

\begin{figure}[t!]
    \centering
        \begin{verbatim}
        (define (domain blocksworld)
            (:requirements :strips :equality)
            (:predicates
                (handEmpty)
                (holding ?x)
                (onTable ?x)
                (on ?x ?y)
                (clear ?x))

           (:action pickup_from_table
              :parameters
                (?x)
              :precondition
                (and (handEmpty)
                     (onTable ?x)
                     (clear ?x))
              :effect
                (and (not (handEmpty))
                     (not (onTable ?x))
                     (holding ?x)))

           (:action putdown_on_stack
              :parameters
                 (?x ?y)
              :precondition
                 (and
                    (not (= ?x ?y))
                    (holding ?x)
                    (clear ?y))
              :effect
                 (and
                    (not (holding ?x))
                    (not (clear ?y))
                    (on ?x ?y)
                    (handEmpty))))
        \end{verbatim}
      \caption{\emph{A fragment of PDDL ``BlocksWorld'' Domain.}}
    \label{fig:pddl-blocksworld}
  \end{figure}

Formalisation of planning algorithms \cite{abdulaziz2019verified}
has shown that even well understood algorithms can produce incorrect plans.
Modern AI planners are complex software artefacts,
and the existing attempts \cite{rizaldi2018formally,penberthy1992ucpop} to verify them only focused on certain aspects of their implementation.
Due to the complexity of  planning problems, many planners will opt for implementations
where efficiency is the primary concern which can further complicate the ability
to formally verify these algorithms. No mainstream planner has been fully verified yet.

AI plan verification seeks to verify plans produced by planners against some
domain model. These tools check properties such as precondition satisfaction,
termination and goal satisfaction to ensure that a plan is valid. For example, PDDL
has a validator \cite{howey2004val} that  performs these checks
and suggests repairs. This more practical and   \emph{lightweight} approach to verification is broadly in line with other lightweight verification trends in the literature~\cite{FLR17}. However, at the same time, it is rather disjoint from
the growing body of research into  \emph{type-based verification}~\cite{DBLP:reference/crypt/Necula11,leroy2012compcert,nanevski2013dependent} 
or \emph{resource logics}~\cite{pym2019resource,calcagno2015moving} that offer more principled, formal and rigorous approaches, as well as richer languages for expressing the verification properties.

In this paper, we are taking an attempt at bridging this gap between the AI planning and the programming language community. We introduce a new formal system inspired
by resource semantics \cite{o2007resources,pym2019resource}, and by the Curry-Howard view on Separation logic
as given e.g. in \cite{N16,polikarpova2019structuring}.
We call the resulting formalism the \emph{proof-carrying plan logic} (or PCP logic for short).
It features:  Hoare triples to describe plans and states;
 the frame rule for local resource-aware reasoning;
 and the  Curry-Howard view on states as types, and state transformations as functions.
The latter feature ensures that plans that we verify in our logic are also executable functions -- which completes
the analogy with the ``proof-carrying code'' research agenda~\cite{DBLP:reference/crypt/Necula11}.

This approach has several advantages over the existing plan verifiers. Firstly, the clear and intuitive formal semantics helps to clarify
the computational properties of AI plans. For example, conditions are embedded into our
rules that ensure the desired property of state consistency is inherent in the logic;
and structural rules of the PCP logic help to clarify the role of constraints in PDDL
and expose some latent properties of AI plans (see Section~\ref{sec:implementation}).
As a result, we were able to prove soundness of  the PCP logic relative to the possible world
semantics as used in the AI planning community, and fully formalise both the logic and the
proof in Agda~\cite{ATH20}. This sets up new standards of rigour for AI plan verification that
is not present in existing planners.

Secondly, some benefits arise as a consequence of adopting a higher level of abstraction.
For example, the existing AI planning verification approaches split into methods for domain and plan verification.
This is potentially harmful, as verifying just one aspect still leaves a gap for bugs and errors.
The PCP logic does not separate  the problem of state consistency and validity of plan execution.
We envisage that the right level of abstraction will enable further extensions to incorporate concurrency and more sophisticated constraints on the states and the plans.

Finally, benefiting from the Curry-Howard approach, our Agda code can  be extracted as verified
executable Haskell or byte code.
We will illustrate all of these concepts by means of an example.

\subsection{Results of this paper by means of an example}\label{sec:ex}

Figures~\ref{fig:pddl-blocksworld2}, \ref{fig:pddl-blocksworld} show the original PDDL syntax for a planning domain and a planning problem.
PDDL will be able to automatically find a plan that satisfies pre- and post-conditions shown in Figure~\ref{fig:pddl-blocksworld2}. In particular, it will find a plan
$$\mathit{f}_{\mathit{ab}} = ((\mathit{pickup\_from\_table\  a})\ ; \ (\mathit{putdown\_on\_stack \  a\  b})).$$  Our goal is to formulate a proof system in which we can (semi-) automatically prove correctness of this plan, given the PDDL domain description, the initial and the goal states.

Looking closer at the domain definition in Figure~\ref{fig:pddl-blocksworld}, we see it declares first-order predicates, as well as actions that operate on pre- and post-conditions. Ignoring temporarily the internal structure of pre- and post-conditions, we can see that
the formalism lends itself naturally to the syntax of Hoare triples~\cite{hoare1969axiomatic}:
$ \{ \mathit{Pre} \} \strans \{ \mathit{Post} \} \ | \  \mathit{action}$. This is our first key intuition.
Somewhat differently from Hoare logic, we see that the domain definition defines
a set of axioms that control actions. So, we will be talking about a certain plan or action possible relative to a domain $\Gamma$. Thus, we will in fact be working with judgements of the form:
$$ \Gamma \vdash \{ \mathit{Pre} \} \strans \{ \mathit{Post} \} \ | \ \mathit{action}$$

Let us now look closer at the structure of the pre- and post-conditions. The domain specification (Figure~\ref{fig:pddl-blocksworld}) states them as formulae of first-order logic (with negation and conjunction), but the problem definition (Figure~\ref{fig:pddl-blocksworld2}) uses lists of atomic propositional formulae to describe initial and goal states. A simple way to resolve this mismatch is to define $\mathit{Pre}$ and $\mathit{Post}$ to be states in which each individual atomic formula is mapped to $+$ or $-$, depending on whether it is considered to be true or false in the state.
For example, $\{ \mathit{onTable}\ \mathit{a}  \mapsto + \} $ is a state with one formula map.
This allows us to formalise the notions of negation, state and (later) action on states.

Of course, we must not allow inconsistent states where a formula is mapped to $+$ and $-$ simultaneously.
To ensure this we introduce a notion of a \emph{valid} state where a state
is valid if it is consistent, i.e. contains no conflicting formula maps.
For example $\{ (\mathit{onTable\ a} \mapsto +) \  * \  (\mathit{onTable\ b} \mapsto +) \} $
is a valid state but  $\{ (\mathit{onTable\ a} \mapsto +) \  * \  (\mathit{onTable\ a} \mapsto -) \} $ is not.

States are not necessarily propositional and in particular
Figure~\ref{fig:pddl-blocksworld} implies first-order language in the domain definitions.
Therefore, we will assume that all  $ \{ \mathit{Pre} \} \strans \{ \mathit{Post} \} \ | \  \mathit{action}$ in $\Gamma$ are implicitly universally quantified as follows:
$\forall \overline{x}.  \{ \mathit{Pre} (\overline{x}) \} \strans \{ \mathit{Post}(\overline{x}) \} \ | \  \mathit{action}(\overline{x}) $.

One final caveat exists. If we look closer at the domain description of  Figure~\ref{fig:pddl-blocksworld}, we will notice that it uses an inequality constraint $x \neq y$ that is not
declared as a domain predicate. We will follow the resource logics tradition~\cite{berdine2005symbolic,polikarpova2019structuring} and separate state descriptions from constraints on states.
That is, we further refine domain descriptions to have the syntax $\phi ; \{ \mathit{Pre} \} \strans \{ \mathit{Post} \} \ | \  \mathit{action}$, where $\phi$ defines constraints on states. For technical reasons, we
formalise $\phi$ to be a list (rather than a conjunction) of constraints, and we use equality ($=$) and inequality ($\neq$) constraints instead of using negation explicitly.
Figure~\ref{fig:context} defines the context $\Gamma_{\mathit{BW}}$ that matches PDDL domain description of Figure~\ref{fig:pddl-blocksworld} in this new language. 

\begin{figure}[t]
    \begin{tabular}{l}

    $[] ;

    \begin{Bmatrix}
    \mathit{handempty} \mapsto +  \\
    *\ \mathit{onTable} \ x \mapsto +  \\
    *\ \mathit{clear} \ x \mapsto +  \\
    \end{Bmatrix}

    \strans

    \begin{Bmatrix}
    \mathit{handEmpty} \mapsto - \\
    * \ \mathit{onTable} \ x\mapsto - \\
    * \ \mathit{holding} \ x \mapsto + \\
    * \ \mathit{clear} \ x \mapsto +
    \end{Bmatrix}

    \ | \ \alpha_1 \ x $ \\

    where $\alpha_1 \equiv \mathit{pickup\_from\_table}$

    \\

    $[x \neq y];
    \begin{Bmatrix}
    \mathit{holding} \ x \mapsto + \\
    * \ \mathit{clear} \ y \mapsto +
    \end{Bmatrix}

    \strans

    \begin{Bmatrix}
      \mathit{holding} \ x \mapsto - \\
      * \ \mathit{clear} \ y \mapsto - \\
      * \ \mathit{on} \ x \ y \mapsto + \\
      * \ \mathit{handEmpty} \mapsto +
    \end{Bmatrix}

    \ | \ \alpha_2 \ x \ y $ \\

    where $\alpha_2 \equiv \mathit{putdown\_on\_stack} $

    \\
    \end{tabular}

    \caption{Context $\Gamma_{\mathit{BW}}$ that defines BlocksWorld PDDL domain of Figure~\ref{fig:pddl-blocksworld}.}
    \label{fig:context}
\end{figure}

It now remains to formulate the rules for the PCP logic. They are very simple: we need a rule ``ApplyAction'' to be able to choose specific action definitions from the context,  we need a rule that composes the actions, similar to the composition rule of Hoare logic~\cite{hoare1969axiomatic},
and we need a frame rule \cite{hayes1981frame,dennett2006cognitive} 
to have local reasoning on states.
Additionally, the system has two structural rules, weakening and shrink.
We show that the rules are sound relative to the possible world semantics of PDDL, and we formalise the PCP logic and the soundness proof in Agda~\cite{ATH20}.

To make use of this Agda library, we can compile the domain and  problem definitions from PDDL directly to Agda code. We then can prove in Agda correctness of the PDDL plans.
For example, we can prove that, given $P_{\mathit{ab}} \equiv \{(\mathit{onTable\ a} \mapsto +) \  * \  (\mathit{onTable\ b} \mapsto +) \ * \ (\mathit{clear \ a} \mapsto +) \ * \ (\mathit{clear \ b} \mapsto +) \ * \ (\mathit{handEmpty} \mapsto +) \}$ and $Q_{\mathit{ab}} \equiv \{(\mathit{on \ a \ b}  \mapsto +) \ * \ (\mathit{on \ Table \ b}  \mapsto +)  \}$ as in Figure~\ref{fig:pddl-blocksworld},
$\Gamma_{\mathit{BW}} \vdash  P_{\mathit{ab}} \strans Q_{\mathit{ab}}   | \ f_{\mathit{ab}}$, i.e. we can certify that
the plan $f_{\mathit{ab}}$ is indeed valid.

Finally, we can take advantage of the Curry-Howard interpretation of $\Gamma_{\mathit{BW}} \vdash  P_{\mathit{ab}} \strans Q_{\mathit{ab}}   | \ f_{\mathit{ab}}$, as ``function $f_{\mathit{ab}}$ has type $P_{\mathit{ab}} \strans Q_{\mathit{ab}}$'' and actually execute $f_{\mathit{ab}}$ as a function. We define an \emph{action handler}, an auxiliary function that executes plans on states.
It will apply the plan $f_{\mathit{ab}}$ to the initial state $P_{\mathit{ab}}$ to obtain the goal state $Q_{\mathit{ab}}$ as a function output.
Moreover, we can extract this code to Haskell or binary files, the latter can be deployed directly on robots, with the advantage of carrying the correctness proof! We show the extracted code for this example and several additional examples in~\cite{ATH20}.

\subsection{The Paper Structure}

The paper proceeds as follows. Section~\ref{grammar}  introduces the PCP logic, proving formally some basic results concerning the ordering and basic operations on states.
Section~\ref{sec:sound} establishes the soundness of the PCP logic and also defines the notion of action handler.
Section~\ref{sec:implementation} describes the implementation~\cite{ATH20}, evaluates it on several benchmark PDDL domains,
and discusses the practical value of using dependent types for implementation of verified plans.
Section~\ref{sec:RW} concludes, and discusses related and future work.

\section{The PCP Logic}\label{grammar}

This section defines the syntax, ordering (subtyping) relation on states,
and the rules of the PCP logic.

\subsection{Syntax of the PCP logic}
We define the PCP syntax in Figure~\ref{fig:grammarsep}.

\textbf{First-order formulas and constraints.}
Let
  $R$ be a set of  predicate symbols $\{R,R_1,R_2,...\}$ with arities,
   $X$ be a set of variables $\{x,x_1,x_2,...\}$,  and
   $C$ be a set of constants $\{c,c_1,c_2,...\}$. Figure~\ref{fig:grammarsep} defines a term as either a variable or a constant.
   An \emph{atomic formula} (or \emph{Atom}) is given by a predicate applied to a finite list of terms.
For example, the atomic formula $\mathit{onTable \ a}$ consists of the predicate
$\mathit{onTable} $ applied to a constant $\mathit{a}$.
This defines the pure first-order part of our logic.
We also distinguish two specific kinds of atomic formulae that feature equality and inequality as predicate symbols, we call these \emph{Constraints}.

We will use abbreviation  $\overline{x}$ to denote a finite list $\{ x_1 , ... \ x_n \}$ of arbitrary length. We will write $R(\overline{x})$ if $R$ contains variables $\overline{x}$.
A substitution is a partial map from $X$ to $C$, and we will use symbols $\{\sigma,\sigma_1,\sigma_2,...\}$
to denote ground substitutions.
Given an atomic formula $R(\overline{x})$ we write $R(\overline{x})[x_i \backslash c_i]$ when we substitute each occurrence
of a variable $x_{i}$ in $\overline{x}$ by a constant $c_{i}$. We say the resulting formula is \emph{ground}, i.e. it contains no variables.

\begin{figure*}[t]
  \centering
  \begin{equation*}
  \begin{aligned}
    \text{Term}
        && \mathit{Term} & \ni t, t_1, ... \ t_n    && :\!:= x \ | \ c
       \\
    \text{Atomic Formulae}
        && \mathit{Atom} & \ni A    && :\!:= R \ (t_1, ... \ t_n)
    \\
    \text{Constraint}
        && \mathit{Constraint} & \ni e  && :\!:=  t = t_1 \ | \ t \neq t_1 \
    \\
    \text{Constraint List}
        && \mathit{CList} & \ni \phi, \psi    && :\!:=  \ [] \ | \  e :: \phi 
 \\
    \text{Actions}
    && \mathit{Act} & \ni a    && :\!:= \alpha \ (t_1, ... \ t_n)
     \\
    \text{Plan}
        && \mathit{Plan} & \ni f, f_1, f_2 && :\!:= \mathit{shrink} \ | \ a \ | \ f ; f_1
        \\
    \text{Polarities}
        && \mathit{Polarity} & \ni z && :\!: = \ + \ |\  -
    \\
    \text{State}
        && \mathit{State} & \ni P, Q, S && :\!:= \mathit{emp} \ | \ A \mapsto z \ | \ P * Q
       \\
    \text{(Planning) Context}
        && \Gamma & \ni \gamma && :\!:= \phi(\overline{x}) ; \{ P(\overline{x})  \}
                            \strans \{ Q(\overline{x}) \} \ | \  \alpha \ \overline{x}
    \\
    \text{Specification}
        && \mathit{Specification} & \ni G && :\!:= \Gamma \vdash \{ P \} \strans \{ Q \} \ | \ f
  \end{aligned}
\end{equation*}
\caption{\emph{\small{ The syntax of PCP logic }}}
  \label{fig:grammarsep}
\end{figure*}

\textbf{Actions and plans.} Let
$\mathrm{A}$ be a set of action names $\{\alpha, \alpha_1, \alpha_2, ...\}$. Figure~\ref{fig:grammarsep} defines an action as an action name applied to a list of terms, e.g.
$ \mathit{pickup\_from\_table \ a}$ is an action. A plan is a sequence of actions; $\mathit{shrink}$ is a special constructor that can be
used in a plan instead of an action, its use will be made clear later.

\textbf{States and contexts.}
Polarities $+$ and $-$ are used to denote absense or presence of certain atomic fact in a world.
Given a polarity $z$,
$A \mapsto z$ is a \emph{formula map}.
A \emph{state}  can be given by an empty state, a formula map or a conjunction of such maps (denoted by $*$).
A state $(A \mapsto z * P)$ is \emph{valid} if $A$ does not occur in $P$ and $P$ is a valid state.
We will only work with valid states in this paper.
A \emph{context} $\Gamma$ contains descriptions of actions in the form
$\phi(\overline{x}) ; \{ P(\overline{x})  \}
                    \strans \{ Q(\overline{x}) \} \ | \  \alpha \ \overline{x}$
where $\{ P(\overline{x})\} \strans \{ Q(\overline{x}) \}$ denotes a transformation
from a state $P(\overline{x})$ to a state $Q(\overline{x})$, $\alpha \ \overline{x}$
is an action and $\phi(\overline{x})$ is a constraint list.

\begin{remark}
  To simplify our notation, we extend the use of notation ``$(\overline{x})$'' from atomic formulae, such as $R(\overline{x})$, to states (e.g. $Q(\overline{x})$), actions (e.g. $\alpha(\overline{x})$) and constraints (e.g. $\phi(\overline{x})$). In all these cases, the presence of $\ \overline{x}$ signifies the presence of free variables $\overline{x}$ in the states, actions, and constraints, respectively.
  We will drop $\overline{x}$ and will write just $Q$, $\alpha$, and $\phi$ to emphasise that the state, action or constraint do not contain any variables, i.e. they are \emph{ground}.
  \end{remark}

A plan specification is a sequent of the form:

\begin{centering}

\begin{tabular}{cc}
    $\Gamma \vdash \{ P \} \strans \{ Q \} \ | \ f$
\end{tabular}

\end{centering}

\begin{figure}[t]
  \centering
  \begin{tabular}{cc}
    \inference[NilSub \ ]
      {}
      {S \st emp}
    &
    \inference[ASub\ ]
      {S' \st S \quad
       A \mapsto z \in S'}
      {S' \st A \mapsto z * S }
  \end{tabular}
  \caption{\emph{Subtyping order on states.}}
  \label{fig:subtyping}
\end{figure}

\noindent It states that given a context $\Gamma$, $f$ is a plan that gives a provable
transformation from (ground) state $P$ to (ground) state $Q$. In the Curry-Howard interpretation of this logic, we view $f$ as a function that inhabits type $ \{ P \} \strans \{ Q \}$.

In all examples, we use the following shorthand notation:
$$R \ t \mapsto z * R \ t_1 \mapsto z \equiv R \ t,t_1 \mapsto z$$

\noindent For example, we will write  $(\mathit{onTable \ a,b} \mapsto +)$ instead of  $(\mathit{onTable \ a} \mapsto +) * (\mathit{onTable \ b} \mapsto +)$.
To emphasise that a formula map binds stronger than $*$, we will put parentheses around formula maps in all examples.
But we will omit the parentheses  in the formal grammar, to keep the notation simple.

  \begin{figure*}[t]
    \centering
    \begin{minipage}{.5\textwidth}
        \centering

        \begin{tabular}{c}
            \inference[Frame\ ]
            {\Gamma \vdash \{ P\}
                \strans \{ Q\} \ | \ \alpha}
                {\Gamma \vdash \{ P * A \mapsto z \}
                    \strans \{ Q * A \mapsto z\} \ | \ \alpha}
        \end{tabular} \\

        \hfill\break

        \begin{tabular}{c}
            \inference[Shrink \ ]
            {Q <: Q' \\
            \Gamma \vdash \{ P\}
                \strans \{ Q \} \ | \ f}
                {\Gamma \vdash \{ P \}
                    \strans \{ Q' \} \ | \ f ; \mathit{shrink}}
        \end{tabular} \\

        \hfill\break

    \end{minipage}%
    \begin{minipage}{0.5\textwidth}
        \centering

        \begin{tabular}{c}
            \inference[Weakening\ ]
            {P' <: P \\
            \Gamma \vdash \{ P\}
                \strans \{ Q \} \ | \ f}
                {\Gamma \vdash \{ P' \}
                    \strans \{ Q \} \ | \ f }
        \end{tabular} \\

        \hfill\break
        \hfill\break
        \hfill\break

        \begin{tabular}{cc}
            \inference[ApplyAction\ ]
                {
                 \phi(\overline{x}) ; \{ P(\overline{x}) \}
                                    \strans \{ Q(\overline{x}) \} \ | \ \alpha (\overline{x}) \in \Gamma }
                {\Gamma \vdash \{ P(\overline{x})[\sigma]  \}
                    \strans \{ Q(\overline{x})[\sigma] \} \ | \ \alpha (\overline{x}) [\sigma]}
        \end{tabular} \\
        \ \ \ \ \ \ \ \ \ \ \ \ \ \ \ \ \ \ \ \ \ \ \
        Where $\phi(\overline{x})[\sigma]$ normalises to $\top$

        \hfill\break

    \end{minipage}

    \begin{tabular}{c}
        \inference[Composition\ ]
        { Q <: Q' \\
          \Gamma \vdash \{ P\}
            \strans \{ Q\} \ | \ f \ \ \ \
                \Gamma \vdash \{ Q'\}
                    \strans \{ R\} \ | \ f_{1}}
            {\Gamma \vdash \{ P\}
                \strans \{ R\} \ | \ f ; f_{1}}
    \end{tabular} \\

    \caption{Rules of the PCP logic. The rules operate on valid states.}
    \label{fig:rules}
\end{figure*}

\subsection{Subtyping (order on states)}

We first recall the subtyping relation and the override operator on states introduced in \cite{schwaab2019proof}, and then establish some lemmas about these,
which will be useful in the later sections. The lemmas have not appeared in \cite{schwaab2019proof}. We omit proofs here, but give them in Agda~\cite{ATH20}.

Figure~\ref{fig:subtyping} defines order $\st$ over states. Following \cite{schwaab2019proof}, we call it subtyping to refer to the fact that states can also be seen as types.
In this paper, subtyping serves us when we need to compare states or decide whether they are equal.
Two states $P$ and $Q$ are considered
equal if $P \st Q \ and \ Q \st P$.

\begin{example}[Subtyping]\label{ex:subt}
  Given:
  $Q \equiv (\{\mathit{onTable\ a} \mapsto -)\ *\ (\mathit{onTable\ b} \mapsto +)\ *\ (\mathit{clear\ a,b} \mapsto +)
\  *\ (\mathit{handEmpty} \mapsto -)\ *\ (\mathit{holding\ a} \mapsto +)$ and
$Q' \equiv (\mathit{onTable\ a} \mapsto -)\ *\ (\mathit{onTable\ b} \mapsto +)\ *\ (\mathit{clear\ a,b} \mapsto +)
\ *\ (\mathit{holding\ a} \mapsto +)$,
we have $Q \st Q'$.
\end{example}

 Subtyping is both
 reflexive and transitive, i.e. it is a \emph{pre-order}.

 \begin{lemma}[Subtyping is Preorder]
   Given states $P, Q, S$, we have:
   \begin{itemize}
   \item (reflexivity) $P \st P$;
     \item (transitivity) $P \st Q$ and $Q \st S$ implies $P \st S$.
     \end{itemize}
   \end{lemma}

In later sections, we will also need an override operator on states:

\begin{definition}[Override Operator \cite{schwaab2019proof}]

\begin{align*}
 P \sqcup emp &= P \\
 P \sqcup [A \mapsto z * Q] &=
   [A \mapsto z * P \backslash \{A \mapsto + * A \mapsto -\}] \sqcup Q
\end{align*}

\end{definition}

The override operator adds all formula maps from one state to the other. If a
mapping for a formula that is to be added already exists, then that formula is removed
before adding the new formula map.

\begin{example}[Override Operator]
  $\ \\ $
  $ (\mathit{handEmpty} \mapsto +) * (\mathit{onTable\ a} \mapsto +) * (\mathit{clear\ a} \mapsto +) \ \sqcup \  \\
  (\mathit{handEmpty} \mapsto -)  * (\mathit{onTable\ a} \mapsto -) * (\mathit{holding\ a} \mapsto +) \\
  = (\mathit{holding\ a} \mapsto +) * (\mathit{onTable\ a} \mapsto -) *  (\mathit{handEmpty} \mapsto -) * (\mathit{clear\ a} \mapsto +)\\ $

\end{example}

We have the following lemmas summarising the properties of the subtyping relation and the override operator.

  \begin{lemma}[Order of Subtyping]\label{lem:ProofSub} Given an atom $A$ and states $P$ and $Q$, if $A \notin Q$
  and $Q \st P$ then $A \notin P$.
  \end{lemma}

  \begin{lemma}[Monotonicity of Subtype Expansion]\label{lem:SubAdd}
  Given states $P$ and $Q$ and a formula map $A \mapsto z$, if $Q \st P$ then
  $A \mapsto z * Q \st P$.
  \end{lemma}

\begin{lemma}[Post-condition Override]\label{lem:PostSubOverride}  $(P \sqcup Q) \st Q$ holds for all states $P$ and $Q$.
\end{lemma}

 \begin{lemma}[Monotonicity of Override]\label{lem:OverrideMembership} Given a polarity $z$, an atom $A$, states
  $P$ and $Q$, if $A \notin Q$ then $A \mapsto z \in (A \mapsto z * P) \sqcup Q$.
  \end{lemma}

\subsection{Normalisation of Constraint Lists}

We will now define a normalisation function for constraint lists. This function takes
a list of constraints and recurses through them checking that they are true.
If a constraint is not true, $\bot$ is returned; otherwise
the empty list case will be reached and $\top$ will be returned.
We use $t \equiv t_1$ to denote syntactic equivalence between terms.

\begin{definition}[Normalisation Function for Constraints]

\begin{align*}
  \mathit{norm}\ [] &= \top \\
  \mathit{norm}\ (t = t_1 :: \phi) &= \boldsymbol{if} \ t \equiv t_1 \ \boldsymbol{then}
            \ \mathit{norm}\ \phi \ \boldsymbol{else} \ \bot \\
  \mathit{norm}\ (t \neq t_1 :: \phi)  &= \boldsymbol{if} \ t \equiv t_1 \ \boldsymbol{then}
            \ \bot \ \boldsymbol{else} \ \mathit{norm}\ \phi \\
\end{align*}

\end{definition}

\begin{example}[Normalisation Function for Constraints]
  We have $\mathit{norm} \ [a = a, b = b] = \top $ but $\mathit{norm} \ [a = a, b = c] = \bot $.
  \end{example}

\subsection{Rules}\label{rules}

 Figure~\ref{fig:rules} gives the rules of the PCP logic.
We will discuss and illustrate each rule in order, using our
running example. In Figure~\ref{fig:pddl-blocksworld} a
PDDL definition of BlocksWorld is defined. An example context $\Gamma_{\mathit{BW}}$, inspired by that definition, is given in Figure~\ref{fig:context}.
Assume that this is the context for all below examples.

  \textbf{ApplyAction}
  checks that an action
is in the context and then constructs the resultant
state given by a ground substitution on that action. For example,
the $\mathit{pickup\_from\_table}$ action is included in $\Gamma_{\mathit{BW}}$ (cf. the first action in Figure~\ref{fig:context}).
Taking $P^a \equiv  \{(\mathit{handempty} \mapsto +) * (\mathit{onTable \ a} \mapsto +) * (\mathit{clear \ a} \mapsto +)\}$ and $Q^a \equiv  \{ (\mathit{handEmpty} \mapsto -) * (\mathit{onTable \ a} \mapsto -) * (\mathit{holding \ a} \mapsto +) * (\mathit{clear \ a} \mapsto +) \}$, we have

\begin{centering}

\begin{mathpar}
     {\inferrule*
         {\ (1) \in \Gamma_{\mathit{BW}} }
         {\Gamma_{\mathit{BW}} \vdash \{ P^a \} \strans
          \{ Q^a \}
          \\ | \ \mathit{pickup\_from\_table \ a }}}
\end{mathpar}
where $(1)$ refers to the first action in  $\Gamma_{\mathit{BW}}$.

\end{centering}

This is the only rule that allows us to access planning domain definitions. Note also that this is the only rule that checks whether constraints on states are satisfied. This is possible thanks to essentially propositional reasoning implemented in planning, thus it is sufficient to check the constraints only once.

\textbf{Composition} rule says that if we have an entailment $\Gamma \vdash \{ P \} \strans \{ Q \} \ | \ f$
we can compose it together with another entailment $\Gamma \vdash \{ Q' \} \strans \{ R \} \ | \ f_1$
to produce $\Gamma \vdash \{ P \} \strans \{ R \} \ | \ f ; f_1$, if $Q \st Q'$.

For this example, we take $Q$ and $Q'$ as in Example~\ref{ex:subt} (with $Q \st Q'$), and we take $P$ and $R$ as follows:\\
$P \equiv (\mathit{onTable\ a,b} \mapsto +)\ *\ (\mathit{clear\ a,b} \mapsto +)\ *\ (\mathit{handEmpty} \mapsto +)$\\
$R \equiv (\mathit{onTable\ a} \mapsto -)\ *\ (\mathit{onTable\ b} \mapsto +)\ *\ (\mathit{clear\ a} \mapsto +)
    \ *\ (\mathit{clear\ b} \mapsto -)\ *\ (\mathit{handEmpty} \mapsto +)\ *\ (\mathit{holding\ a} \mapsto -)
    \ *\ (\mathit{on\ a\ b} \mapsto +)$.
    Abbreviating $\mathit{putdown\_on\_stack \ a \ b}$ as $\alpha$ and $\mathit{pickup\_from\_table\ a}$ as $f$, we have the following application of the Composition rule:

\begin{centering}

\begin{mathpar}

{\inferrule* 
{ Q <: Q' \\
  \Gamma_{\mathit{BW}} \vdash \{ P\}
    \strans \ \{ Q\} \ | \ f \ \ \ \
        \Gamma_{\mathit{BW}} \vdash \{ Q'\}
            \strans \ \{ R\} \ | \ \alpha}
    {\Gamma_{\mathit{BW}} \vdash \{ P\}
        \strans \ \{ R\} \ | \ f ; \alpha}}

\end{mathpar}

\end{centering}

  The \textbf{Frame} rule allows the addition of formula maps to both states in an entailment,
provided the atom of the formula map does not already have a mapping in
either state. Continuing the derivation in one of the previous examples,
 the following application of the frame rule is possible:

\begin{centering}

\begin{mathpar}
     {\inferrule* [Right=Frame,rightstyle=\bf,
                       leftskip=2em,rightskip=2em]
         {\Gamma_{\mathit{BW}} \vdash \{ P^a \} \strans \{ Q^a \} \ | \ \mathit{pickup\_from\_table \ a}}
         {\Gamma_{\mathit{BW}} \vdash \{ P^a * (\mathit{onTable \ b} \mapsto +) \}
          \strans \\ \{ Q^a * (\mathit{onTable \ b} \mapsto +) \} \ | \ \mathit{pickup\_from\_table \ a}}}
\end{mathpar}

\end{centering}

In our system, the frame rule is more restrictive than can be seen in other
logics such as Separation logic \cite{reynolds2002separation,berdine2005symbolic},
as it can only be used at an action level but
not at a plan level.
The following example shows the problem with consistency of derivations,  if we apply
the frame rule to arbitrary judgements of the form
$\Gamma \vdash \{ P \} \strans \{ Q \} \ | \ f$.

\begin{example}[Problems with the Frame rule for complex plans]
Imagine we have an action $\alpha$ with the
transformation $\{\mathit{clear\ a} \mapsto +\} \strans \{(\mathit{clear\ a} \mapsto -) \ * \ (\mathit{clear\ b} \mapsto +) \}$
and another action $\alpha'$ with the transformation $\{\mathit{clear\ a} \mapsto -\} \strans \{\mathit{clear\ a} \mapsto + \}$
then we can compose these two actions together to generate the entailment:
$ \Gamma \vdash \{\mathit{clear\ a} \mapsto + \} \strans \{\mathit{clear\ a} \mapsto +\} \ | \  \alpha ; \alpha' $.
We have lost the information $\mathit{clear\ b} \mapsto +$ and if the Frame rule was
not bound to single actions we could frame incorrectly in the entailment:
$ \Gamma \vdash \{(\mathit{clear\ a} \mapsto +) * (\mathit{clear\ b} \mapsto -) \} \strans \{(\mathit{clear\ a} \mapsto +) * (\mathit{clear\ b} \mapsto -)\} \ | \  \alpha ; \alpha' $, getting a derivation inconsistent with the action definition.
\end{example}

If we want to apply this rule on judgements involving complex plans instead of single actions,
then we would need to ensure that the framed atom is not mapped in any state at any level in the
plan derivation. This could be done by amending the restrictions on the frame rule
or by amending the other rules to prevent loss of information.

 \textbf{Weakening}
is applied before composition, when a formula map we want in the precondition $P$
already exists in the (previously obtained) post-condition $Q$.
The above example shows a use case with the action $\alpha$ as defined above: 

\begin{centering}

\begin{mathpar}
       {\inferrule* 
           {\Gamma_{\mathit{BW}} \vdash \{ (\mathit{clear\ b} \mapsto +) \ * \ (\mathit{holding\ a} \mapsto +) \} \\
            \strans \{ (\mathit{clear\ b} \mapsto -) \ * \ (\mathit{holding\ a} \mapsto -)
              \ * \ (\mathit{on\ a\ b} \mapsto +) \ * \ \\  (\mathit{handEmpty} \mapsto +) \} \ | \ \alpha}
            {\Gamma_{\mathit{BW}} \vdash \{ (\mathit{clear\ b} \mapsto +) \ * \ (\mathit{holding\ a} \mapsto +) \ * \\ (\mathit{handEmpty} \mapsto -) \}
             \strans \{ (\mathit{clear\ b} \mapsto -) \ * \\ (\mathit{holding\ a} \mapsto -)
               \ * \ (\mathit{on\ a\ b} \mapsto +) \ * \  (\mathit{handEmpty} \mapsto +) \} \ | \ \alpha}}
\end{mathpar}

\end{centering}

\noindent In BlocksWorld it is implied that $\mathit{handEmpty}$ is false when $holding$
any block is true and vice versa. This leads the $\mathit{putdown\_on\_stack}$ action's
preconditions to only contain the precondition that $\mathit{holding}$ a block has to be true
and we use weakening to gain back the information that $\mathit{handEmpty}$ is false.

\textbf{Shrink} allows us to shrink and reorder the post-condition state. Any
postcondition state $Q$ can be replaced with $Q'$ as long as it is
a subtype of the current post state. Shrink can appear anywhere in a plan but
currently the main use of this rule is when  we have a goal state that is smaller
than the obtained post-condition, for example:

\begin{centering}

\begin{mathpar}
       {\inferrule* 
           {\Gamma \vdash \{ P \}
            \strans  \{ (\mathit{clear\ b} \mapsto -) \ * \ (\mathit{holding\ a} \mapsto -)
              \ * \\ (\mathit{on\ a\ b} \mapsto +) \ * \ (\mathit{handEmpty} \mapsto +) \} \ | \ f}
            {\Gamma \vdash \{P \}
             \strans \{ (\mathit{on\ a\ b} \mapsto +) \} \ | \ f;\mathit{shrink}}}
\end{mathpar}

\end{centering}

Frame, Weakening and Shrink are structural rules, i.e. they do not change the computational properties of plans, and with the exception of Shrink, do not change the plans syntactically.
We finish this section by stating two lemmas that explain subtyping for plans derived by structural rules. Note that all actions have unique definitions in any given context $\Gamma$. The proofs of these lemmas are given in Agda~\cite{ATH20}.

\begin{lemma}[Property of structural rules (left)]\label{lem:PreSatP} If there is a derivation for
  $ \Gamma \vdash \{ P \} \strans \{ Q  \} \ | \ \alpha $ by the rules of  Figure~\ref{fig:rules}
  we have: $\{ P'(\overline{x}) \} \strans \{ Q'(\overline{x})  \} \ | \ \alpha(\overline{x}) \in \Gamma$
   and $P \st P'(\overline{x})[\sigma]$.
  \end{lemma}
  \begin{lemma}[Property of structural rules (right)]\label{lem:PostSatQ} If there is a derivation
  $ \Gamma \vdash \{ P \} \strans \{ Q  \} \ | \ \alpha $  by the rules of  Figure~\ref{fig:rules}, we have   $\{ P'(\overline{x}) \} \strans \{ Q'(\overline{x})  \} \ | \ \alpha(\overline{x}) \in \Gamma$
  and $Q \st Q'(\overline{x})[\sigma] $.
\end{lemma}

Given a planning context $\Gamma$, we say that a plan $f$ is \emph{well-typed} (for $ \{ P \} \strans \{ Q \}$), if there is a derivation of   $\Gamma \vdash \{ P \} \strans \{ Q \} \ | \ f$
by the rules of Figure~\ref{fig:rules}.

\section{Soundness of the PCP Logic}\label{sec:sound}

We now show that the PCP logic we introduced in previous sections is sound relative to the possible world semantics of PDDL~\cite{fox2003pddl2}.

\subsection{Possible World Semantics for PDDL Languages}

Coming back to our running example of a PDDL domain, given in Figure~\ref{fig:pddl-blocksworld}, we notice that it is defined in a subset of first-order logic,
while the actual problem description (Figure~\ref{fig:pddl-blocksworld2}) contains only ground terms.
This motivates us to formally define \emph{PDDL formulae} as follows:

\begin{definition}[PDDL Formulae]
  \begin{equation*}
  \begin{aligned}
       \text{Ground Atoms}
        && \mathit{GAtom} & \ni A^g    && :\!:= R \ (c_1, ... \ c_n)
    \\
    \text{PDDL Formulae}
        && \mathit{Form} & \ni F,F_1 ... \ F_n    && :\!:=  A^g \ | \ \neg A^g \ | \ F \dot\wedge F_1\\
  \end{aligned}
  \end{equation*}
\end{definition}

\begin{figure}[t]
  \centering
  \begin{tabular}{cc}
    \inference
      {w \models_z  F \quad
       w \models_z  F_1}
      {w \models_z F \land F_1}
    &
    \inference
      {w \models_{\minus z} A^g}
      {w \models_z \neg A^g}
\\
\\
    \inference
      {A^g \in w}
      {w \models_\mplus A^g}
    &
    \inference
      {A^g \not \in w}
      {w \models_{\minus} A^g}
  \end{tabular}
  \caption{\emph{\small{Declarative interpretation of PDDL formulae. We define $\minus z$ by taking $ \minus + =  \minus$ and $ \minus \minus = +$.}}}
  \label{fig:well-formed-world}
\end{figure}

 \emph{Possible world semantics for PDDL}~\cite{fox2003pddl2} is defined in Figure~\ref{fig:well-formed-world}. 
 A possible world, or just a \emph{world} is a set of ground atomic formulae.
We use letter $w$ to denote a single possible world.
Given a world $w$, a PDDL formula \emph{$F$ is satisfied by $w$} if
$w \models_{+} F$ can be derived by the rules of
Figure~\ref{fig:well-formed-world}. It should be noted that negation can only be
applied to atomic formulae.

It is useful to establish a correspondence between states and formulae.
Following \cite{schwaab2019proof,SHFK18}, we achieve this by introducing a ``normalisation'' function from PDDL formulae to states.

\begin{definition}[Normalisation of PDDL Formulae to States\cite{schwaab2019proof}]
The function $ \tnorm{z}{}$ \emph{normalises}  a PDDL formula to a state:
\begin{align*}
  \tnorm{z}{(F \land F_1)} S &= \tnorm{z}{F_1}\tnorm{z}{F} S \\
  \tnorm{z}{\neg A^g} S &= \tnorm{\minus z}{A^g} S \\
  \tnorm{z}{A^g} S &= A^g \mapsto z * S
\end{align*}
We write $\tnorm{z}{F}$ to mean ($\tnorm{z}{F} \mathit{emp}$).

\end{definition}

\begin{example}[Normalisation of a Formula to a State]

  {\small $$\pnorm{
  (\mathit{handEmpty} \land \neg \mathit{onTable\ a})} =
  \mathit{handEmpty} \mapsto + \ * \ \mathit{onTable\ a} \mapsto \minus$$}
\end{example}

Normalisation is sound relative to the possible world semantics.
A world $w_S$ is a \emph{well-formed world} for a given state $S$,
if the world $w_S$ contains all $A^g$'s such that $(A^g \mapsto {+}) \in S$ and
contains no $A^g$'s such that $(A^g \mapsto {\minus})  \in S
$\footnote{By abuse of notation that will not cause confusion, we will
use the symbol $\in$ to denote State membership, as well as set membership.}.
Generally $w_S$ is not uniquely defined, and we use the notation $\langle w_S \rangle$ to refer
to the (necessarily finite) set of all $w_S$.

\begin{example}[Well-Formed Worlds] \hfill\break
  If $S =$ {\small $\pnorm{ (\mathit{handEmpty} \land \neg \mathit{onTable\ a})} $}, then
    $w_S$ may be given by e.g. $w_1 = ${\small $ \{ \mathit{handEmpty} \} $, or $w_2
    = \{ \mathit{handEmpty}, \mathit{onTable\ b} \} $}, or any other world
    containing  {\small $ \mathit{handEmpty} $} but not  {\small $ \mathit{onTable\ a}$}. The given
    formula will be satisfied by any such $w_S$.
\end{example}

Well-formed worlds have the following property :

\begin{lemma}[Subtyping and Well-Formed Worlds]\label{lem:SubResp}
  If we have states $P$ and $Q$, $Q \st P$ and $w \in \langle w_Q \rangle$ then $w \in \langle w_P \rangle$.
\end{lemma}

Finally, we prove that normalisation is sound and complete:

\begin{theorem}[Soundness and completeness of normalisation \cite{schwaab2019proof,SHFK18}]\label{thm:norm-soundness}
  Given a formula $F$ and a world $w$, it holds that $ w \models_z F$ iff $w \in \langle w_{\tnorm{z}{F}} \rangle$.
\end{theorem}
\begin{proof}
$(\Rightarrow)$ is proven by induction on the derivation of $w \models_z F$.
$(\Leftarrow)$ follows by induction on the shape of $F$, cf. the attached Agda file \cite{ATH20} for the fully formalised proof.

\end{proof}

\subsection{Soundness Theorem}

We want to show that if we derive $\Gamma \vdash \{ \tnorm{z}{F} \} \strans \{ \tnorm{z}{F_1} \} \ | \ f$
using the rules given in Figure~\ref{fig:rules} then we are guaranteed that the \emph{evaluation}
of the plan $f$ on a world that satisfies $F$ produces a new world satisfying $F_1$.

To evaluate a plan we will define an \emph{evaluation function $\denote{.}$}
 that will interpret actions on worlds.
Recall that every state $S$ maps to a world $ w_s$.
Let us use notation $\delta$ for an arbitrary mapping (an \emph{action handler})
 that maps each action $\phi ; \{ P \} \strans \{ Q \} \ | \ \alpha$ to insertions and
 deletions on the world $w_S$ according to $\alpha$'s action on $S$.
 We then define the evaluation function $(\denote{f}^{\delta} \ w)$ that
\emph{evaluates} a plan $f$ in a world $w$ using an action handler $\delta$:
\begin{definition}[Evaluation Function]
\begin{align*}
  \denote{\mathit{shrink}}^{\delta} \ w &= w \\
  \denote{a}^{\delta} \ w &= \delta \ a \ w \\
  \denote{f ; f_1}^{\delta} \ w &=  \denote{f_1}^{\delta} (\denote{f}^{\delta} \ w)
\end{align*}
\end{definition}

The evaluation function has three cases. The shrink case just returns the world itself,
 as there is no computational meaning for a shrink action in evaluation. For a single
 action, evaluation applies the action handler to the world. For a complex plan, evaluation recurses to sub-plans.

 The following property of action handlers will be used in the soundness proof:
  \begin{lemma}[Action Handler Strengthening]\label{lem:Strengthening} If $( \delta \ \alpha \ w ) \in \langle w_{Q} \rangle$
  and $( \delta \ \alpha \ w ) \in \langle w_{A \mapsto z} \rangle$ then
  $( \delta \ \alpha \ w ) \in \langle w_{A \mapsto z * Q} \rangle$.
  \end{lemma}

We now proceed to define the notion of a well-formed handler,
that will be used to prove soundness of the PCP logic.

\begin{definition}[Well-Formed Handler]

We say that an action handler $\delta$ is \emph{well-formed} if, given:
\begin{itemize}
\item a
  context $\Gamma$ with  $\phi(\overline{x}) ; \{ P'(\overline{x}) \} \strans \{ Q(\overline{x}) \} \ | \ \alpha(\overline{x}) \ \in \ \Gamma$,
\item  a state $P$,
  such that $P(x) \st P'(\overline{x})[\sigma]$ for some ground substitution $\sigma$ and $\phi(\overline{x})[\sigma]$ normalises to $\top$,
  \item a world $w \in \langle w_P \rangle$,
\end{itemize}
$\delta$ satisfies the following property: $( \delta \ (\alpha(\overline{x})[\sigma]) \ w ) \in \langle w_{P \sqcup Q(\overline{x})[\sigma]} \rangle$.

\end{definition}

The next two theorems show that executing a well-typed plan $f$ by the evaluation function
$\denote{f}^{\delta} \ w$ is  \emph{sound}, for any well-formed handler $\delta$.

\begin{theorem}[Soundness of evaluation for normalized formulae]\label{thm:eval-soundness}
  Suppose $\Gamma \vdash \{ P \} \strans \{ Q \} \ | \ f$. Then for any
  $w \in \langle w_P \rangle$, and any
  well-formed handler $\delta$, it follows that $\denote{f}^\delta \ w \in \langle w_Q \rangle$.
\end{theorem}
\begin{proof}
  The proof proceeds by structural induction on the typing derivation
  $\Gamma \vdash \{ P \} \strans \{ Q \} \ | \ f$.
%
%
%
In each of the below cases, we take $P$, $w \in \langle w_{P} \rangle$, and assume  $\Gamma \vdash \{ P \} \strans \{ Q \} |\ f$
was proven by application of a given rule in Figure~\ref{fig:rules},  and in each case  we will aim to show that
$\denote{f}^\delta \ w \in \langle w_Q \rangle$.

  \begin{base}[$\mathrm{ApplyAction}$]

  Suppose we have a proof for $\Gamma \vdash \{ P \} \strans \{ Q \} |\ f$ by
  means of the rule ApplyAction. The rules premise requires
  that some $\phi(\overline{x}) ; \{ P'(\overline{x}) \} \strans \{ Q'(\overline{x}) \} \ | \ \alpha(\overline{x}) \ \in \ \Gamma$,
  and moreover there exists $\sigma \ s.t.\  P'(\overline{x})[\sigma] \equiv P,\
  Q'(\overline{x})[\sigma] \equiv Q,\ \alpha(\overline{x})[\sigma] \equiv \ f$ and $\phi(\overline{x})[\sigma]$
  normalises to $\top$.

  Because $\delta$ is well-formed and   $w \in \langle w_P \rangle$, we have:
  $( \delta \ f \ w ) \in \langle w_{P \sqcup Q} \rangle$. We note that $P \st P'(\overline{x})[\sigma] $
  because $P \equiv P'(\overline{x})[\sigma]$ by the conditions of the rule, and $P \st P$ by reflexivity of subtyping relation.

  It remains to show that $( \delta \ f \ w ) \in \langle w_{P \sqcup Q} \rangle$
  implies that $( \delta \ f \ w ) \in \langle w_{Q} \rangle$. We know that
  $(P \sqcup Q) \st Q$ from Lemma \ref{lem:PostSubOverride} and can therefore deduce
  $( \delta \ a \ w ) \in \langle w_{Q} \rangle$ by applying Lemma
  \ref{lem:SubResp}.

\end{base}

\begin{case}[$\mathrm{Weakening}$]
Taking $P$, $w \in \langle w_{P} \rangle$ as before, we assume  $\Gamma \vdash \{ P \} \strans \{ Q \} |\ f$ was proven by applying Weakening.
By inductive hypothesis we know that
there is a proof of \ $\Gamma \vdash \{ P' \} \strans \{ Q \} |\ f$, such that $P \st P'$ and
$\denote{f}^\delta \ w' \in \langle w_Q \rangle$ if $w' \in \langle w_{P'} \rangle$ for some $w'$.
By Lemma~\ref{lem:SubResp}, we know that $w \in \langle w_P \rangle$ implies $w \in \langle w_{P'} \rangle$.
And so we have $\denote{f}^\delta \ w \in \langle w_Q \rangle$ as required.

\end{case}
\begin{case}[$\mathrm{Shrink}$]
  We now assume that
$\Gamma \vdash \{ P \} \strans \{ Q \} |\ f$ is obtained by application of Shrink, i.e. $f \equiv (f_1; shrink)$ for some $f_1$. By inductive hypothesis we know that
there is a proof of \ $\Gamma \vdash \{ P \} \strans \{ Q' \} |\ f_1$ such that $Q' \st Q$, and $\denote{f_1}^\delta \ w \in \langle w_{Q'} \rangle$ if $w \in \langle w_P \rangle$.
Because we already have  $w \in \langle w_P \rangle$ among our assumptions, we get $\denote{f_1}^\delta \ w \in \langle w_{Q'} \rangle$.
We apply Lemma~\ref{lem:SubResp} to get $\denote{f_1}^\delta \ w \in \langle w_Q \rangle$. It remains to show that $\denote{f_1;shrink}^\delta \ w \in \langle w_Q \rangle$.
By definition of the evaluation function, $\denote{f_1;shrink}^\delta \ w = \denote{f_1}^\delta \ (\denote{shrink}^\delta \ w) = \denote{f_1}^\delta \ w$, as required.

\end{case}
\begin{case}[$\mathrm{Composition}$]
 We now assume that
 $\Gamma \vdash \{ P \} \strans \{ Q \} |\ f$ by application of Composition.
 By inductive hypothesis we know that, for some $f_1$ and $f_2$ such that $f \equiv f_1;f_2$, and for some $Q'$ and $Q''$ such that $Q' \st Q''$,
 \begin{itemize}
 \item there is a proof of $\Gamma \vdash \{ P \} \strans \{ Q' \} |\ f_1$  and $\denote{f_1}^\delta \ w \in \langle w_{Q'} \rangle$ if $w \in \langle w_P \rangle$;
  \item there is a proof of $\Gamma \vdash \{ Q'' \} \strans \{ Q \} |\ f_2$  and $\denote{f_2}^\delta \ w' \in \langle w_Q \rangle$ if $w' \in \langle w_{Q''} \rangle$;
\end{itemize}
Because we already have  $w \in \langle w_P \rangle$ among our assumptions, we get $\denote{f_1}^\delta \ w \in \langle w_{Q'} \rangle$.
Next, we need to apply  Lemma~\ref{lem:SubResp} and the fact that $Q \st Q''$ to get  $\denote{f_1}^\delta \ w \in \langle w_{Q''} \rangle$.
Thus we found a suitable  $w' \equiv \denote{f_1}^\delta \ w$.
But then we get  $\denote{f_2}^\delta ({\denote{f_1}^\delta \ w}) \in \langle w_Q \rangle$.
Finally, by definition of the evaluation function, we know that
$\denote{f_1;f_2}^\delta \ w =  \denote{f_2}^{\delta} ({\denote{f_1}^{\delta} \ w})$. And so we get $\denote{f_1;f_2}^\delta \ w \in \langle w_Q \rangle$.
\end{case}
\begin{case}[$\mathrm{Frame}$]
   We now assume that
   $\Gamma \vdash \{ P \} \strans \{ Q \} |\ f$ by application of the Frame rule, that is, $f \equiv \alpha$,
   $P \equiv (P' * A \mapsto z)$, $Q \equiv (Q' * A \mapsto z)$ (for some $\alpha$, $P'$, $Q'$, $A$ and $z$), moreover $w \in \langle w_{A \mapsto z * P'} \rangle$,
   $A \notin P'$, $A \notin Q'$. By the inductive hypothesis, we know that there is a proof of \
   $\Gamma \vdash \{ P' \} \strans \{ Q' \} |\ \alpha$  and $\denote{\alpha}^\delta \ w' \in \langle w_{Q'} \rangle$ if $w' \in \langle w_{P'} \rangle$, for any $w'$.

   By  Lemma~\ref{lem:SubAdd} and the fact that $P' \st P'$, we get  $A \mapsto z * P' \st P'$. We then use Lemma~\ref{lem:SubResp}, and our assumption $w \in \langle w_{P} \rangle$ to assert that $w \in \langle w_{P'} \rangle$, and therefore we get
   $\denote{\alpha}^\delta \ w \in \langle w_{Q'} \rangle$.  It remains to show that $\denote{\alpha}^\delta \ w \in \langle w_{Q} \rangle$.

   By the definition of evaluation function, $\denote{\alpha}^\delta \ w =  \delta \ \alpha \ w $.
Lemma~\ref{lem:Strengthening} lets us combine two
  results: 1. $( \delta \ a \ w ) \in \langle w_{Q'} \rangle$ and
  2. $( \delta \  \alpha \ w ) \in \langle w_{A \mapsto z} \rangle$ to produce the goal
  $( \delta \  \alpha \ w ) \in \langle w_{A \mapsto z * Q'} \rangle$ which gives us  $( \delta \  \alpha \ w ) \in \langle w_{Q} \rangle$ and therefore
  $\denote{\alpha}^\delta \ w \in \langle w_{Q} \rangle$, as required.

  It only remains to show that $( \delta \  \alpha \ w ) \in \langle w_{A \mapsto z} \rangle$.
  To prove this, we will use the fact that $\delta$ is a well-formed handler, and consider $( \delta \  \alpha \ w )$.
  Recall that
\begin{itemize}
  \item $w \in \langle w_{A \mapsto z * P'} \rangle$, and,
  \item by inductive hypothesis, there is a derivation for   $\Gamma \vdash \{ P' \} \strans \{ Q' \} |\ \alpha$.
  Therefore, there is
  $\phi(\overline{x}) ; \{ P''(\overline{x}) \} \strans \{ Q''(\overline{x}) \} \ | \ \alpha(\overline{x}) \ \in \ \Gamma$ by Lemma~\ref{lem:PreSatP}.
\item Also by Lemma~\ref{lem:PreSatP}, we have $P' \st P''(\overline{x})[\sigma]$, for some $\sigma$.
  \item We know that  $\alpha(\overline{x})[\sigma]$ must normalise to $\top$, or there would be no derivation for  $\Gamma \vdash \{ P' \} \strans \{ Q' \} |\ \alpha$.
\end{itemize}
  Given these four  conditions, a well-formed handler must satisfy the property: $( \delta \  \alpha \ w ) \in \langle w_{(A \mapsto z * P') \sqcup Q''(\overline{x})[\sigma]} \rangle$.
  We can apply Lemma~\ref{lem:SubResp} and show that  $( \delta \  \alpha \ w ) \in \langle w_{A \mapsto z} \rangle$, if we can show that
  $(A \mapsto z * P') \sqcup Q''(\overline{x})[\sigma] \st (A \mapsto z)$.
  Using  Lemma~\ref{lem:OverrideMembership}
   we can establish that $A \mapsto z \in (A \mapsto z * P') \sqcup Q''(\overline{x})[\sigma]$ if $A \notin Q''(\overline{x})[\sigma]$.
To show  $A \notin Q''(\overline{x})[\sigma]$, we use  Lemma~\ref{lem:ProofSub}, Lemma~\ref{lem:PostSatQ} (which gives us $Q' \st Q''(\overline{x})[\sigma]$) and the assumption that  $A \notin Q'$. From  $A \mapsto z \in (A \mapsto z * P') \sqcup Q''(\overline{x})[\sigma]$ we obtain  $(A \mapsto z * P') \sqcup Q''(\overline{x})[\sigma] \st (A \mapsto z) $ by using the subtyping derivation rules.
  \end{case}
\end{proof}

\begin{theorem}[Soundness of Evaluation]\label{thm:eval-soundness2}
  Suppose $\Gamma \vdash \pnorm{F_1} \strans \pnorm{F_2} |\ f$ then for any $w$ such
  that $w \models_\mplus F_1$, and any well-formed $\delta$ it follows
  $\denote{f}^\delta \ w \models_\mplus F_2$.
\end{theorem}
\begin{proof}
  By assumption $w \models_\mplus F_1$ and  by Theorem~\ref{thm:norm-soundness}, we have $w \in \langle w_{\pnorm{F_1}} \rangle$.
  Then from Theorem~\ref{thm:eval-soundness},
  we have $\denote{f}^\delta \ w \in \langle w_{\pnorm{F_2}} \rangle$. Thus by Theorem~\ref{thm:norm-soundness}, we obtain $\denote{f}^\delta \ w \models_\mplus F_2$.
\end{proof}

\noindent Thus if $f$  is well-typed, we are guaranteed that the execution
of $f$ in world $w$ is correct.

\section{Implementation and Evaluation}\label{sec:implementation}
As mentioned already, the PCP logic and all lemmas and theorems presented in this paper are formalised in Agda, see~\cite{ATH20}.
This gives us assurance of the correcteness of the presented approach.
This Agda module also serves as a \emph{standard library} for verifying PDDL plans. Recall that
in Section~\ref{sec:ex}, for example, our task was to verify an exact plan $f_{\mathit{ab}}$, i.e. to derive $\Gamma_{\mathit{BW}} \vdash  P_{\mathit{ab}} \strans Q_{\mathit{ab}}   | \ f_{\mathit{ab}}$.
To do this, we need to create
an additional file that defines $\Gamma_{BW}$, $P_{\mathit{ab}}$, $Q_{\mathit{ab}}$ and $f_{\mathit{ab}}$ in Agda syntax.
 Then we need to construct a proof in Agda that this plan is indeed valid.
That is, judgements like  $\Gamma_{\mathit{BW}} \vdash  P_{\mathit{ab}} \strans Q_{\mathit{ab}}   | \ f_{\mathit{ab}}$  are not automatically type-checked by Agda,
but require manual proofs (using the rules of the PCP logic, cf. Figure~\ref{fig:rules}).

To mitigate this, we automate the following two tasks:
\begin{enumerate}
\item conversion from PDDL syntax to Agda.

  For example the $\mathit{pickup\_from\_table\ }x$ action as given in Figure~\ref{fig:pddl-blocksworld} is
translated to the following snippet of Agda code:

\begin{code}
\>[0]\AgdaFunction{Γ₁}\AgdaSpace{}%
\AgdaSymbol{(}\AgdaOperator{\AgdaInductiveConstructor{pickup\AgdaUnderscore{}from\AgdaUnderscore{}table}}\AgdaSpace{}%
\AgdaBound{x}\AgdaSymbol{)}\AgdaSpace{}%
\AgdaSymbol{=}\AgdaSpace{}%
\AgdaInductiveConstructor{[]}\AgdaSpace{}%
\AgdaOperator{\AgdaInductiveConstructor{,}}\<%
\\
\>[0][@{}l@{\AgdaIndent{0}}]%
\>[1]\AgdaSymbol{(}\AgdaInductiveConstructor{+}\AgdaSpace{}%
\AgdaOperator{\AgdaInductiveConstructor{,}}\AgdaSpace{}%
\AgdaInductiveConstructor{handempty}\AgdaSymbol{)}\AgdaSpace{}%
\AgdaOperator{\AgdaInductiveConstructor{*}}\AgdaSpace{}%
\AgdaSymbol{(}\AgdaInductiveConstructor{+}\AgdaSpace{}%
\AgdaOperator{\AgdaInductiveConstructor{,}}\AgdaSpace{}%
\AgdaInductiveConstructor{ontable}\AgdaSpace{}%
\AgdaBound{x}\AgdaSymbol{)}\AgdaSpace{}%
\AgdaOperator{\AgdaInductiveConstructor{*}}\AgdaSpace{}%
\AgdaSymbol{(}\AgdaInductiveConstructor{+}\AgdaSpace{}%
\AgdaOperator{\AgdaInductiveConstructor{,}}\AgdaSpace{}%
\AgdaInductiveConstructor{clear}\AgdaSpace{}%
\AgdaBound{x}\AgdaSymbol{)}\AgdaSpace{}%
\AgdaOperator{\AgdaInductiveConstructor{*}}\AgdaSpace{}%
\AgdaInductiveConstructor{[]}\AgdaSpace{}%
\AgdaOperator{\AgdaInductiveConstructor{,}}\<%
\\
\>[1]\AgdaSymbol{(}\AgdaInductiveConstructor{+}\AgdaSpace{}%
\AgdaOperator{\AgdaInductiveConstructor{,}}\AgdaSpace{}%
\AgdaInductiveConstructor{clear}\AgdaSpace{}%
\AgdaBound{x}\AgdaSymbol{)}\AgdaSpace{}%
\AgdaOperator{\AgdaInductiveConstructor{*}}\AgdaSpace{}%
\AgdaSymbol{(}\AgdaInductiveConstructor{-}\AgdaSpace{}%
\AgdaOperator{\AgdaInductiveConstructor{,}}\AgdaSpace{}%
\AgdaInductiveConstructor{handempty}\AgdaSymbol{)}\AgdaSpace{}%
\AgdaOperator{\AgdaInductiveConstructor{*}}\<%
\\
\>[1]\AgdaSymbol{(}\AgdaInductiveConstructor{-}\AgdaSpace{}%
\AgdaOperator{\AgdaInductiveConstructor{,}}\AgdaSpace{}%
\AgdaInductiveConstructor{ontable}\AgdaSpace{}%
\AgdaBound{x}\AgdaSymbol{)}\AgdaSpace{}%
\AgdaOperator{\AgdaInductiveConstructor{*}}\AgdaSpace{}%
\AgdaSymbol{(}\AgdaInductiveConstructor{+}\AgdaSpace{}%
\AgdaOperator{\AgdaInductiveConstructor{,}}\AgdaSpace{}%
\AgdaInductiveConstructor{holding}\AgdaSpace{}%
\AgdaBound{x}\AgdaSymbol{)}\AgdaSpace{}%
\AgdaOperator{\AgdaInductiveConstructor{*}}\AgdaSpace{}%
\AgdaInductiveConstructor{[]}\<%
\\
\end{code}

\item  PCP logic proof generation for Agda, given a PDDL plan.

Figure \ref{fig:derivation} shows an example of the Agda proof for $\Gamma_{\mathit{BW}} \vdash  P_{\mathit{ab}} \strans Q_{\mathit{ab}}   | \ f_{\mathit{ab}}$ generated automatically given the domain and planning problem specifications of  Figures~\ref{fig:pddl-blocksworld2} and~\ref{fig:pddl-blocksworld}.
\end{enumerate}

\begin{figure*}
\begin{code}
\>[0]\AgdaFunction{P}\AgdaSpace{}%
\AgdaSymbol{=}\AgdaSpace{}%
\AgdaSymbol{(}\AgdaInductiveConstructor{+}\AgdaSpace{}%
\AgdaOperator{\AgdaInductiveConstructor{,}}\AgdaSpace{}%
\AgdaSymbol{(}\AgdaInductiveConstructor{ontable}\AgdaSpace{}%
\AgdaInductiveConstructor{a}\AgdaSymbol{))}\AgdaSpace{}%
\AgdaOperator{\AgdaInductiveConstructor{*}}\AgdaSpace{}%
\AgdaSymbol{(}\AgdaInductiveConstructor{+}\AgdaSpace{}%
\AgdaOperator{\AgdaInductiveConstructor{,}}\AgdaSpace{}%
\AgdaSymbol{(}\AgdaInductiveConstructor{ontable}\AgdaSpace{}%
\AgdaInductiveConstructor{b}\AgdaSymbol{))}\AgdaSpace{}%
\AgdaOperator{\AgdaInductiveConstructor{*}}\AgdaSpace{}%
\AgdaSymbol{(}\AgdaInductiveConstructor{+}\AgdaSpace{}%
\AgdaOperator{\AgdaInductiveConstructor{,}}\AgdaSpace{}%
\AgdaSymbol{(}\AgdaInductiveConstructor{clear}\AgdaSpace{}%
\AgdaInductiveConstructor{a}\AgdaSymbol{))}\AgdaSpace{}%
\AgdaOperator{\AgdaInductiveConstructor{*}}\AgdaSpace{}%
\AgdaSymbol{(}\AgdaInductiveConstructor{+}\AgdaSpace{}%
\AgdaOperator{\AgdaInductiveConstructor{,}}\AgdaSpace{}%
\AgdaSymbol{(}\AgdaInductiveConstructor{clear}\AgdaSpace{}%
\AgdaInductiveConstructor{b}\AgdaSymbol{))}\AgdaSpace{}%
\AgdaOperator{\AgdaInductiveConstructor{*}}\AgdaSpace{}%
\AgdaSymbol{(}\AgdaInductiveConstructor{+}\AgdaSpace{}%
\AgdaOperator{\AgdaInductiveConstructor{,}}\AgdaSpace{}%
\AgdaSymbol{(}\AgdaInductiveConstructor{handempty}\AgdaSymbol{))}\AgdaSpace{}%
\AgdaOperator{\AgdaInductiveConstructor{*}}\AgdaSpace{}%
\AgdaInductiveConstructor{[]}\<%
\\
\\[\AgdaEmptyExtraSkip]%
\>[0]\AgdaFunction{Q}\AgdaSpace{}%
\AgdaSymbol{:}\AgdaSpace{}%
\AgdaFunction{State}\<%
\\
\>[0]\AgdaFunction{Q}\AgdaSpace{}%
\AgdaSymbol{=}\AgdaSpace{}%
\AgdaSymbol{(}\AgdaInductiveConstructor{+}\AgdaSpace{}%
\AgdaOperator{\AgdaInductiveConstructor{,}}\AgdaSpace{}%
\AgdaSymbol{(}\AgdaInductiveConstructor{on}\AgdaSpace{}%
\AgdaInductiveConstructor{a}\AgdaSpace{}%
\AgdaInductiveConstructor{b}\AgdaSymbol{))}\AgdaSpace{}%
\AgdaOperator{\AgdaInductiveConstructor{*}}\AgdaSpace{}%
\AgdaSymbol{(}\AgdaInductiveConstructor{+}\AgdaSpace{}%
\AgdaOperator{\AgdaInductiveConstructor{,}}\AgdaSpace{}%
\AgdaSymbol{(}\AgdaInductiveConstructor{ontable}\AgdaSpace{}%
\AgdaInductiveConstructor{b}\AgdaSymbol{))}\AgdaSpace{}%
\AgdaOperator{\AgdaInductiveConstructor{*}}\AgdaSpace{}%
\AgdaInductiveConstructor{[]}\<%
\\
\\[\AgdaEmptyExtraSkip]%
\>[0]\AgdaFunction{P'}\AgdaSpace{}%
\AgdaSymbol{:}\AgdaSpace{}%
\AgdaFunction{State}\<%
\\
\>[0]\AgdaFunction{P'}\AgdaSpace{}%
\AgdaSymbol{=}\AgdaSpace{}%
\AgdaSymbol{(}\AgdaInductiveConstructor{+}\AgdaSpace{}%
\AgdaOperator{\AgdaInductiveConstructor{,}}\AgdaSpace{}%
\AgdaInductiveConstructor{ontable}\AgdaSpace{}%
\AgdaInductiveConstructor{b}\AgdaSymbol{)}\AgdaSpace{}%
\AgdaOperator{\AgdaInductiveConstructor{*}}\AgdaSpace{}%
\AgdaSymbol{(}\AgdaInductiveConstructor{+}\AgdaSpace{}%
\AgdaOperator{\AgdaInductiveConstructor{,}}\AgdaSpace{}%
\AgdaInductiveConstructor{clear}\AgdaSpace{}%
\AgdaInductiveConstructor{b}\AgdaSymbol{)}\AgdaSpace{}%
\AgdaOperator{\AgdaInductiveConstructor{*}}\AgdaSpace{}%
\AgdaSymbol{(}\AgdaInductiveConstructor{+}\AgdaSpace{}%
\AgdaOperator{\AgdaInductiveConstructor{,}}\AgdaSpace{}%
\AgdaInductiveConstructor{handempty}\AgdaSymbol{)}\AgdaSpace{}%
\AgdaOperator{\AgdaInductiveConstructor{*}}\AgdaSpace{}%
\AgdaSymbol{(}\AgdaInductiveConstructor{+}\AgdaSpace{}%
\AgdaOperator{\AgdaInductiveConstructor{,}}\AgdaSpace{}%
\AgdaInductiveConstructor{ontable}\AgdaSpace{}%
\AgdaInductiveConstructor{a}\AgdaSymbol{)}\AgdaSpace{}%
\AgdaOperator{\AgdaInductiveConstructor{*}}\AgdaSpace{}%
\AgdaSymbol{(}\AgdaInductiveConstructor{+}\AgdaSpace{}%
\AgdaOperator{\AgdaInductiveConstructor{,}}\AgdaSpace{}%
\AgdaInductiveConstructor{clear}\AgdaSpace{}%
\AgdaInductiveConstructor{a}\AgdaSymbol{)}\AgdaSpace{}%
\AgdaOperator{\AgdaInductiveConstructor{*}}\AgdaSpace{}%
\AgdaInductiveConstructor{[]}\<%
\\
\\[\AgdaEmptyExtraSkip]%
\>[0]\AgdaFunction{plan}\AgdaSpace{}%
\AgdaSymbol{:}\AgdaSpace{}%
\AgdaDatatype{f}\<%
\\
\>[0]\AgdaFunction{plan}\AgdaSpace{}%
\AgdaSymbol{=}\AgdaSpace{}%
\AgdaSymbol{(}\AgdaInductiveConstructor{join}\AgdaSpace{}%
\AgdaSymbol{(}\AgdaInductiveConstructor{join}\AgdaSpace{}%
\AgdaSymbol{(}\AgdaInductiveConstructor{act}\AgdaSpace{}%
\AgdaSymbol{(}\AgdaOperator{\AgdaInductiveConstructor{pickup\AgdaUnderscore{}from\AgdaUnderscore{}table}}%
\>[44]\AgdaInductiveConstructor{a}\AgdaSymbol{))}\AgdaSpace{}%
\AgdaSymbol{(}\AgdaInductiveConstructor{act}\AgdaSpace{}%
\AgdaSymbol{(}\AgdaOperator{\AgdaInductiveConstructor{putdown\AgdaUnderscore{}on\AgdaUnderscore{}stack}}%
\>[72]\AgdaInductiveConstructor{a}\AgdaSpace{}%
\AgdaInductiveConstructor{b}\AgdaSymbol{)))}\AgdaSpace{}%
\AgdaInductiveConstructor{shrink}\AgdaSymbol{)}\<%
\\
\\
\>[0]\AgdaFunction{Derivation}\AgdaSpace{}%
\AgdaSymbol{:}\AgdaSpace{}%
\AgdaFunction{Γ₁}\AgdaSpace{}%
\AgdaOperator{\AgdaDatatype{,}}\AgdaSpace{}%
\AgdaFunction{P}\AgdaSpace{}%
\AgdaOperator{\AgdaInductiveConstructor{$\strans$}}\AgdaSpace{}%
\AgdaFunction{Q}\AgdaSpace{}%
\AgdaOperator{\AgdaDatatype{¦}}\AgdaSpace{}%
\AgdaFunction{plan}\<%
\\
\>[0]\AgdaFunction{Derivation}\AgdaSpace{}%
\AgdaSymbol{=}\AgdaSpace{}%
\AgdaInductiveConstructor{weakening}\AgdaSpace{}%
\AgdaFunction{P}\AgdaSpace{}%
\AgdaSymbol{(}\AgdaFunction{from-yes}\AgdaSpace{}%
\AgdaSymbol{(}\AgdaFunction{decSub}\AgdaSpace{}%
\AgdaFunction{P'}\AgdaSpace{}%
\AgdaFunction{P}\AgdaSymbol{))}\AgdaSpace{}%
\AgdaInductiveConstructor{tt}\AgdaSpace{}%
\AgdaSymbol{(}\AgdaInductiveConstructor{shrink}\AgdaSpace{}%
\AgdaFunction{Q}\AgdaSpace{}%
\AgdaInductiveConstructor{tt}\<%
\\
\>[0][@{}l@{\AgdaIndent{0}}]%
\>[8]\AgdaSymbol{(}\AgdaFunction{from-yes}\AgdaSpace{}%
\AgdaSymbol{(}\AgdaFunction{decSub}\AgdaSpace{}%
\AgdaFunction{Q}\AgdaSpace{}%
\AgdaSymbol{((}\AgdaInductiveConstructor{-}%
\>[396I]\AgdaOperator{\AgdaInductiveConstructor{,}}\AgdaSpace{}%
\AgdaInductiveConstructor{ontable}\AgdaSpace{}%
\AgdaInductiveConstructor{a}\AgdaSymbol{)}\AgdaSpace{}%
\AgdaOperator{\AgdaInductiveConstructor{*}}\AgdaSpace{}%
\AgdaSymbol{(}\AgdaInductiveConstructor{+}\AgdaSpace{}%
\AgdaOperator{\AgdaInductiveConstructor{,}}\AgdaSpace{}%
\AgdaInductiveConstructor{ontable}\AgdaSpace{}%
\AgdaInductiveConstructor{b}\AgdaSymbol{)}\AgdaSpace{}%
\AgdaOperator{\AgdaInductiveConstructor{*}}\AgdaSpace{}%
\AgdaSymbol{(}\AgdaInductiveConstructor{+}\AgdaSpace{}%
\AgdaOperator{\AgdaInductiveConstructor{,}}\AgdaSpace{}%
\AgdaInductiveConstructor{clear}\AgdaSpace{}%
\AgdaInductiveConstructor{a}\AgdaSymbol{)}\AgdaSpace{}%
\AgdaOperator{\AgdaInductiveConstructor{*}}\AgdaSpace{}%
\AgdaSymbol{(}\AgdaInductiveConstructor{-}\AgdaSpace{}%
\AgdaOperator{\AgdaInductiveConstructor{,}}\AgdaSpace{}%
\AgdaInductiveConstructor{holding}\AgdaSpace{}%
\AgdaInductiveConstructor{a}\AgdaSymbol{)}\AgdaSpace{}%
\AgdaOperator{\AgdaInductiveConstructor{*}}\<%
\\
\>[.][@{}l@{}]\<[396I]%
\>[32]\AgdaSymbol{(}\AgdaInductiveConstructor{-}\AgdaSpace{}%
\AgdaOperator{\AgdaInductiveConstructor{,}}\AgdaSpace{}%
\AgdaInductiveConstructor{clear}\AgdaSpace{}%
\AgdaInductiveConstructor{b}\AgdaSymbol{)}\AgdaSpace{}%
\AgdaOperator{\AgdaInductiveConstructor{*}}\AgdaSpace{}%
\AgdaSymbol{(}\AgdaInductiveConstructor{+}\AgdaSpace{}%
\AgdaOperator{\AgdaInductiveConstructor{,}}\AgdaSpace{}%
\AgdaInductiveConstructor{on}\AgdaSpace{}%
\AgdaInductiveConstructor{a}\AgdaSpace{}%
\AgdaInductiveConstructor{b}\AgdaSymbol{)}\AgdaSpace{}%
\AgdaOperator{\AgdaInductiveConstructor{*}}\AgdaSpace{}%
\AgdaSymbol{(}\AgdaInductiveConstructor{+}\AgdaSpace{}%
\AgdaOperator{\AgdaInductiveConstructor{,}}\AgdaSpace{}%
\AgdaInductiveConstructor{handempty}\AgdaSymbol{)}\AgdaSpace{}%
\AgdaOperator{\AgdaInductiveConstructor{*}}\AgdaSpace{}%
\AgdaInductiveConstructor{[]}\AgdaSymbol{)))}\<%
\\
\>[8]\AgdaSymbol{(}\AgdaInductiveConstructor{composition}\AgdaSpace{}%
\AgdaSymbol{(}\AgdaFunction{from-yes}\AgdaSpace{}%
\AgdaSymbol{(}\AgdaFunction{decSub}\<%
\\
\>[8][@{}l@{\AgdaIndent{0}}]%
\>[12]\AgdaSymbol{((}\AgdaInductiveConstructor{-}\AgdaSpace{}%
\AgdaOperator{\AgdaInductiveConstructor{,}}\AgdaSpace{}%
\AgdaInductiveConstructor{ontable}\AgdaSpace{}%
\AgdaInductiveConstructor{a}\AgdaSymbol{)}\AgdaSpace{}%
\AgdaOperator{\AgdaInductiveConstructor{*}}\AgdaSpace{}%
\AgdaSymbol{(}\AgdaInductiveConstructor{+}\AgdaSpace{}%
\AgdaOperator{\AgdaInductiveConstructor{,}}\AgdaSpace{}%
\AgdaInductiveConstructor{ontable}\AgdaSpace{}%
\AgdaInductiveConstructor{b}\AgdaSymbol{)}\AgdaSpace{}%
\AgdaOperator{\AgdaInductiveConstructor{*}}\AgdaSpace{}%
\AgdaSymbol{(}\AgdaInductiveConstructor{+}\AgdaSpace{}%
\AgdaOperator{\AgdaInductiveConstructor{,}}\AgdaSpace{}%
\AgdaInductiveConstructor{clear}\AgdaSpace{}%
\AgdaInductiveConstructor{a}\AgdaSymbol{)}\AgdaSpace{}%
\AgdaOperator{\AgdaInductiveConstructor{*}}\AgdaSpace{}%
\AgdaSymbol{(}\AgdaInductiveConstructor{+}\AgdaSpace{}%
\AgdaOperator{\AgdaInductiveConstructor{,}}\AgdaSpace{}%
\AgdaInductiveConstructor{holding}\AgdaSpace{}%
\AgdaInductiveConstructor{a}\AgdaSymbol{)}\AgdaSpace{}%
\AgdaOperator{\AgdaInductiveConstructor{*}}\AgdaSpace{}%
\AgdaSymbol{(}\AgdaInductiveConstructor{+}\AgdaSpace{}%
\AgdaOperator{\AgdaInductiveConstructor{,}}\AgdaSpace{}%
\AgdaInductiveConstructor{clear}\AgdaSpace{}%
\AgdaInductiveConstructor{b}\AgdaSymbol{)}\AgdaSpace{}%
\AgdaOperator{\AgdaInductiveConstructor{*}}\AgdaSpace{}%
\AgdaInductiveConstructor{[]}\AgdaSymbol{)}\<%
\\
\>[12]\AgdaSymbol{((}\AgdaInductiveConstructor{+}\AgdaSpace{}%
\AgdaOperator{\AgdaInductiveConstructor{,}}\AgdaSpace{}%
\AgdaInductiveConstructor{ontable}\AgdaSpace{}%
\AgdaInductiveConstructor{b}\AgdaSymbol{)}\AgdaSpace{}%
\AgdaOperator{\AgdaInductiveConstructor{*}}\AgdaSpace{}%
\AgdaSymbol{(}\AgdaInductiveConstructor{+}\AgdaSpace{}%
\AgdaOperator{\AgdaInductiveConstructor{,}}\AgdaSpace{}%
\AgdaInductiveConstructor{clear}\AgdaSpace{}%
\AgdaInductiveConstructor{b}\AgdaSymbol{)}\AgdaSpace{}%
\AgdaOperator{\AgdaInductiveConstructor{*}}\AgdaSpace{}%
\AgdaSymbol{(}\AgdaInductiveConstructor{+}\AgdaSpace{}%
\AgdaOperator{\AgdaInductiveConstructor{,}}\AgdaSpace{}%
\AgdaInductiveConstructor{clear}\AgdaSpace{}%
\AgdaInductiveConstructor{a}\AgdaSymbol{)}\AgdaSpace{}%
\AgdaOperator{\AgdaInductiveConstructor{*}}\AgdaSpace{}%
\AgdaSymbol{(}\AgdaInductiveConstructor{-}\AgdaSpace{}%
\AgdaOperator{\AgdaInductiveConstructor{,}}\AgdaSpace{}%
\AgdaInductiveConstructor{handempty}\AgdaSymbol{)}\AgdaSpace{}%
\AgdaOperator{\AgdaInductiveConstructor{*}}\AgdaSpace{}%
\AgdaSymbol{(}\AgdaInductiveConstructor{-}\AgdaSpace{}%
\AgdaOperator{\AgdaInductiveConstructor{,}}\AgdaSpace{}%
\AgdaInductiveConstructor{ontable}\AgdaSpace{}%
\AgdaInductiveConstructor{a}\AgdaSymbol{)}\AgdaSpace{}%
\AgdaOperator{\AgdaInductiveConstructor{*}}\AgdaSpace{}%
\AgdaSymbol{(}\AgdaInductiveConstructor{+}\AgdaSpace{}%
\AgdaOperator{\AgdaInductiveConstructor{,}}\AgdaSpace{}%
\AgdaInductiveConstructor{holding}\AgdaSpace{}%
\AgdaInductiveConstructor{a}\AgdaSymbol{)}\AgdaSpace{}%
\AgdaOperator{\AgdaInductiveConstructor{*}}\AgdaSpace{}%
\AgdaInductiveConstructor{[]}\AgdaSymbol{)))}\<%
\\
\>[0]\<%
\\
\>[8]\AgdaSymbol{((}\AgdaInductiveConstructor{frame}\AgdaSpace{}%
\AgdaInductiveConstructor{+}\AgdaSpace{}%
\AgdaSymbol{(}\AgdaInductiveConstructor{ontable}\AgdaSpace{}%
\AgdaInductiveConstructor{b}\AgdaSymbol{)}\AgdaSpace{}%
\AgdaSymbol{(λ}\AgdaSpace{}%
\AgdaBound{z}\AgdaSpace{}%
\AgdaSymbol{→}\AgdaSpace{}%
\AgdaBound{z}\AgdaSymbol{)}\AgdaSpace{}%
\AgdaSymbol{(λ}\AgdaSpace{}%
\AgdaBound{z}\AgdaSpace{}%
\AgdaSymbol{→}\AgdaSpace{}%
\AgdaBound{z}\AgdaSymbol{)}\<%
\\
\>[8][@{}l@{\AgdaIndent{0}}]%
\>[10]\AgdaSymbol{(}\AgdaInductiveConstructor{frame}\AgdaSpace{}%
\AgdaInductiveConstructor{+}\AgdaSpace{}%
\AgdaSymbol{(}\AgdaInductiveConstructor{clear}\AgdaSpace{}%
\AgdaInductiveConstructor{b}\AgdaSymbol{)}\AgdaSpace{}%
\AgdaSymbol{(λ}\AgdaSpace{}%
\AgdaBound{z}\AgdaSpace{}%
\AgdaSymbol{→}\AgdaSpace{}%
\AgdaBound{z}\AgdaSymbol{)}\AgdaSpace{}%
\AgdaSymbol{(λ}\AgdaSpace{}%
\AgdaBound{z}\AgdaSpace{}%
\AgdaSymbol{→}\AgdaSpace{}%
\AgdaBound{z}\AgdaSymbol{)}\AgdaSpace{}%
\AgdaSymbol{(}\AgdaInductiveConstructor{applyAction}\AgdaSpace{}%
\AgdaInductiveConstructor{tt}\AgdaSpace{}%
\AgdaInductiveConstructor{tt}\AgdaSpace{}%
\AgdaInductiveConstructor{tt}\AgdaSymbol{))))}\<%
\\
\>[8]\AgdaSymbol{((}\AgdaInductiveConstructor{frame}\AgdaSpace{}%
\AgdaInductiveConstructor{-}\AgdaSpace{}%
\AgdaSymbol{(}\AgdaInductiveConstructor{ontable}\AgdaSpace{}%
\AgdaInductiveConstructor{a}\AgdaSymbol{)}\AgdaSpace{}%
\AgdaSymbol{(λ}\AgdaSpace{}%
\AgdaBound{z}\AgdaSpace{}%
\AgdaSymbol{→}\AgdaSpace{}%
\AgdaBound{z}\AgdaSymbol{)}\AgdaSpace{}%
\AgdaSymbol{(λ}\AgdaSpace{}%
\AgdaBound{z}\AgdaSpace{}%
\AgdaSymbol{→}\AgdaSpace{}%
\AgdaBound{z}\AgdaSymbol{)}\<%
\\
\>[8][@{}l@{\AgdaIndent{0}}]%
\>[10]\AgdaSymbol{(}\AgdaInductiveConstructor{frame}\AgdaSpace{}%
\AgdaInductiveConstructor{+}\AgdaSpace{}%
\AgdaSymbol{(}\AgdaInductiveConstructor{ontable}\AgdaSpace{}%
\AgdaInductiveConstructor{b}\AgdaSymbol{)}\AgdaSpace{}%
\AgdaSymbol{(λ}\AgdaSpace{}%
\AgdaBound{z}\AgdaSpace{}%
\AgdaSymbol{→}\AgdaSpace{}%
\AgdaBound{z}\AgdaSymbol{)}\AgdaSpace{}%
\AgdaSymbol{(λ}\AgdaSpace{}%
\AgdaBound{z}\AgdaSpace{}%
\AgdaSymbol{→}\AgdaSpace{}%
\AgdaBound{z}\AgdaSymbol{)}\<%
\\
\>[10]\AgdaSymbol{(}\AgdaInductiveConstructor{frame}\AgdaSpace{}%
\AgdaInductiveConstructor{+}\AgdaSpace{}%
\AgdaSymbol{(}\AgdaInductiveConstructor{clear}\AgdaSpace{}%
\AgdaInductiveConstructor{a}\AgdaSymbol{)}\AgdaSpace{}%
\AgdaSymbol{(λ}\AgdaSpace{}%
\AgdaBound{z}\AgdaSpace{}%
\AgdaSymbol{→}\AgdaSpace{}%
\AgdaBound{z}\AgdaSymbol{)}\AgdaSpace{}%
\AgdaSymbol{(λ}\AgdaSpace{}%
\AgdaBound{z}\AgdaSpace{}%
\AgdaSymbol{→}\AgdaSpace{}%
\AgdaBound{z}\AgdaSymbol{)}\AgdaSpace{}%
\AgdaSymbol{(}\AgdaInductiveConstructor{applyAction}\AgdaSpace{}%
\AgdaInductiveConstructor{tt}\AgdaSpace{}%
\AgdaInductiveConstructor{tt}\AgdaSpace{}%
\AgdaInductiveConstructor{tt}\AgdaSymbol{))))))}\<%
\\
\end{code}
\caption{Agda typing derivation for BlocksWorld problem and domain given
in Figures~\ref{fig:pddl-blocksworld2} and~\ref{fig:pddl-blocksworld}.}
\label{fig:derivation}
\end{figure*}

For the first task, we
translate a given planning problem and domain to a single Agda file. Conversion of \emph{objects} is
one-to-one, i.e. the list of \emph{objects} is given as a list of constructors for the datatype $C$ that stores constants.
To convert states, we change the list syntax from Lisp style to
Agda style and add the relevant polarity. \emph{Predicates} and \emph{actions} are translated
to Agda by representing them as functions from constants to the relevant type.
For example the predicate $(\mathit{on}\ ?x\ ?y)$ is translated to
$\mathit{on}: \ C \rightarrow C \rightarrow R$.
Action descriptions are described by a
parametrised \emph{precondition} and \emph{effect} list in PDDL as shown in Figure \ref{fig:pddl-blocksworld} .
The PDDL \emph{precondition} list contains \emph{constraints} and formulas which are separate
in our context so we separate them when translating to Agda. Preconditions,
constraints and effects are then mapped one-to-one into a context description.

For the second task, we
 implemented a solver for generating Agda type derivations,
given a plan. A high level overview of the solver algorithm is shown in Figure \ref{fig:automation}.

We thus obtain a parser and a proof generator (implemented in Lisp\footnote{Both PDDL and Emacs are written in Lisp, which determined our choice for using Lisp here.})
that can process plans given in PDDL. However, we delegate the proof-checking
(as type-checking) to Agda. This latter step ultimately ensures fully formal plan verification.
We call the resulting tool \emph{plan verifier}.
The actual implementation \cite{ATH20} contains further instructions and examples.

\subsection{Evaluation of the Library Performance}

Table~\ref{tab:eval} shows the results of evaluating the plan verifier
over a few benchmark PDDL domains: BlocksWorld, Logistics, Satellite
and Mprime\cite{planningbenchmarks}. All domains use the STRIPS
requirement with Mprime also requiring equality and negative preconditions. All of
these examples are generated automatically by supplying a plan and the PDDL domain description to the plan verifier. 

\begin{table}[t]
\centering
 \begin{tabular}{|p{1.7cm}| p{1.7cm} p{1.7cm} p{1.7cm}|}
 \hline
PDDL Domain & Plan Length (number of actions) & Proof Generation Time (seconds)   & Typechecking Time (seconds) \\
 \hline\hline
 Blocksworld & 10 & 0.03 & 10.33 \\
 \hline
 Logistics & 24 & 0.07 & 28.86 \\
 \hline
 Satellite & 9 & 0.03 & 15.66 \\
 \hline
 Mprime & 11 & 0.09 & 42.02 \\
 \hline
\end{tabular}
\caption{Evaluation of the plan verifier. All tests were performed on an Intel Core i5-4670K processor with 8GB of RAM.}\label{tab:eval}
\end{table}

This evaluation shows that our system scales from BlocksWorld to more complicated
domains, even with increasing plan length.
Firstly, it helps to off-load time-consuming plan search to STRIPS. Transforming plans into Agda proofs does not take long (cf. middle column of Table~\ref{tab:eval}).
Type-checking time may look worrying, however there is plenty of room for improving it.
The given type-checking times reflect the fact that our Lisp script generates excessively long  Agda proofs.
This happens because we frame all formulas in when we generate the proofs (see Figure~\ref{fig:automation}).
For example, for the Logistics domain, we have a PDDL plan of length 24. For it, we have PCP/Agda proof with nearly 900 rule applications.
Yet, looking closer, we can see that they are mostly frame rules (838 Frame rules, 23 Compositions, 24 ApplyAction, 1 Weakening, 1 shrink).

Similarly, for Mprime example, we have a PDDL plan of length 11, but nearly 900 rule applications in PCP/Agda. Once again, most of them are applications
of the frame rule: (865 frame rules, 10 compositions, 11 ApplyAction, 1 weak and 1 shrink).

Ignoring the redundant applications of the frame rule, we see linear dependency of the PCP/Agda proof size relative to STRIPS plan size.
Thus, we believe that the type checking time shown in Table~\ref{tab:eval} does not point to limitations of a type-based approach, but merely reflects the
inefficiency of the Lisp script that generates Agda proofs. We will address these shortcomings in future work.

\subsection{Leveraging the Power of Dependent Types}\label{sec:dependent}

Agda is of course also a dependently-typed language. And, as we mentioned in the introduction,
the benefit of this approach is the ability to use plans as functions (using the action handlers).
One benefit of this would be easy extensions to
practical scenarios in which dependent types impose further restrictions and checks on action handlers.
Action handlers currently have
a type $\mathit{Action} \rightarrow \mathit{World} \rightarrow \mathit{World}$. Within a dependently-typed setting, it is easy
to extend this type with say
an energy constraint that limits the number of actions
that can be taken. Assume a scenario when a robot is given a certain amount of energy, or ``fuel'', and
each action execution consumes one unit of energy; the robot may not consume more energy than given.
It only takes a few lines of code to add this information to our current implementation:

\begin{example}[Action Handler Energy Consumption]

\begin{code}
\\
\\[\AgdaEmptyExtraSkip]%
\>[0]\AgdaKeyword{data}\AgdaSpace{}%
\AgdaDatatype{Energy}\AgdaSpace{}%
\AgdaSymbol{:}\AgdaSpace{}%
\AgdaDatatype{Nat}\AgdaSpace{}%
\AgdaSymbol{->}\AgdaSpace{}%
\AgdaPrimitiveType{Set}\AgdaSpace{}%
\AgdaKeyword{where}\<%
\\
\>[0][@{}l@{\AgdaIndent{0}}]%
\>[2]\AgdaInductiveConstructor{en}\AgdaSpace{}%
\AgdaSymbol{:}\AgdaSpace{}%
\AgdaSymbol{(}\AgdaBound{n}\AgdaSpace{}%
\AgdaSymbol{:}\AgdaSpace{}%
\AgdaDatatype{Nat}\AgdaSymbol{)}\AgdaSpace{}%
\AgdaSymbol{->}\AgdaSpace{}%
\AgdaDatatype{Energy}\AgdaSpace{}%
\AgdaBound{n}\<%
\\
\\[\AgdaEmptyExtraSkip]%
\>[0]\AgdaFunction{EnergyValue}\AgdaSpace{}%
\AgdaSymbol{:}\AgdaSpace{}%
\AgdaSymbol{∀}\AgdaSpace{}%
\AgdaSymbol{\{}\AgdaBound{n}\AgdaSymbol{\}}\AgdaSpace{}%
\AgdaSymbol{->}\AgdaSpace{}%
\AgdaDatatype{Energy}\AgdaSpace{}%
\AgdaBound{n}\AgdaSpace{}%
\AgdaSymbol{->}\AgdaSpace{}%
\AgdaDatatype{Nat}\<%
\\
\>[0]\AgdaFunction{EnergyValue}\AgdaSpace{}%
\AgdaSymbol{\{}\AgdaBound{n}\AgdaSymbol{\}}\AgdaSpace{}%
\AgdaBound{x}\AgdaSpace{}%
\AgdaSymbol{=}\AgdaSpace{}%
\AgdaBound{n}\<%
\\
\\[\AgdaEmptyExtraSkip]%
\>[0]\AgdaFunction{actionHandler}\AgdaSpace{}%
\AgdaSymbol{:}\AgdaSpace{}%
\AgdaPrimitiveType{Set}\<%
\\
\>[0]\AgdaFunction{actionHandler}\AgdaSpace{}%
\AgdaSymbol{=}\AgdaSpace{}%
\AgdaSymbol{∀}\AgdaSpace{}%
\AgdaSymbol{\{}\AgdaBound{n}\AgdaSymbol{\}}\AgdaSpace{}%
\AgdaSymbol{->}%
\>[26]\AgdaBound{Action}%
\>[4741I]\AgdaSymbol{->}\AgdaSpace{}%
\AgdaFunction{World}\AgdaSpace{}%
\AgdaOperator{\AgdaFunction{×}}\AgdaSpace{}%
\AgdaDatatype{Energy}\AgdaSpace{}%
\AgdaSymbol{(}\AgdaInductiveConstructor{suc}\AgdaSpace{}%
\AgdaBound{n}\AgdaSymbol{)}\<%
\\
\>[.][@{}l@{}]\<[4741I]%
\>[33]\AgdaSymbol{->}\AgdaSpace{}%
\AgdaFunction{World}\AgdaSpace{}%
\AgdaOperator{\AgdaFunction{×}}\AgdaSpace{}%
\AgdaDatatype{Energy}\AgdaSpace{}%
\AgdaBound{n}\<%
\\
\>[0]\<%
\end{code}
\end{example}

Now our implementation incorporates constraints on energy consumption, and the action execution will be bound by the amount of the given energy.
This is a really powerful way to use dependent types as it improves readability,
and also provides endless possibilities for incorporating various computational constraints in the plan execution.

\begin{figure}

\begin{tcolorbox}
Given a PDDL domain $D$, a PDDL problem description $P_D$, and a plan $f_1$,
the Agda proof script for $f_1$ in PCP logic is generated as follows:
\begin{enumerate}
\item Parse $D$, $P_D$ into Lisp syntax:
\begin{enumerate}
\item Store the objects in $D$, initial and goal world in $P_D$ as variables in Lisp. The initial world id stored in a variable $w$ representing the current world.
\item Convert actions from $D$ into parametrised Lisp functions that generate preconditions and post-conditions.
\end{enumerate}
\item For all actions in the plan $f_1$:
\begin{enumerate}
\item Store $w$ in a backup variable so that the Agda subtyping relations can be generated.
\item Use the relevant Lisp function (as defined in (1.b)) to generate the preconditions and postconditions of the current action.
\item Generate the frame axioms by comparing the preconditions of the action to $w$
where all formula maps in $w$ that are not in the preconditions are framed in.
\item Use Lisp functions to apply the action to the world and store the result in the world variable.
\end{enumerate}
\item Use all stored results to generate and write Agda proof to file:
\begin{enumerate}
\item Start derivation with the Weakening rule to allow for the reordering of the initial state.
\item Use the Shrink rule over the rest of the derivation to shrink the result to
the goal state.
\item The rest of the derivation proceeds by composing all actions in the plan $f_1$
with the relevant subtyping relations and frame axioms.
\end{enumerate}
\item Typecheck the generated Agda file to confirm the validity of the proof.
\end{enumerate}
\end{tcolorbox}

\caption{Overview of the code that automatically generates PCP proofs in Agda given PDDL Domain and plan. The code is given in~\cite{ATH20}.}
\label{fig:automation}

\end{figure}

\subsection{Extraction of Plans  to Executable Code}

We can go one step further, and use Agda's code extraction library and compile our verified plans
into executable Haskell programs or executable
byte code. The latter may be deployed directly in robots.
The process is fully automated by existing Agda libraries, and subsequent execution of the byte code takes just seconds.
For example,  we have
compiled the BlocksWorld, Logistics and Satellite examples into byte code where
all examples run
in just 0.02 seconds. We refer the reader to~\cite{ATH20} for further details.

\subsection{Lessons Learnt: Effects and States}\label{sec:stateorder}

As it often happens with verification projects, this work helps to uncover some
previously unknown or unnoticed  properties of PDDL.
We will give two examples here.

As seen in Figure~\ref{fig:pddl-blocksworld}, the syntax of PDDL defines actions by ``preconditions''
and ``effects''. The PCP logic formalises both as states. Yet, there is a subtle difference between an effect and a state.
Recall that an effect is executed by
deleting all false formula maps from a state and adding all true formula maps.
To convert an effect to a state,
we must keep the list of all unaffected formula maps intact.
Also, as we have shown, the states come with the notion of ordering, but effects do not.
These simple  observations have surprisingly powerful consequences.

\textbf{Example 1. Ordering and consistency}.
 Take the action $move$ from the Logistic domain:

\begin{tabular}{c}

$  [] ;

\begin{Bmatrix}
\mathit{isVehicle}\ v \mapsto +  \\
*\ \mathit{isLocation} \ loc1 \mapsto +  \\
*\ \mathit{isLocation} \ loc2 \mapsto +  \\
*\ \mathit{isAt}\ v\ loc1 \mapsto +
\end{Bmatrix}

\strans

\begin{Bmatrix}
\mathit{isVehicl}e\ v \mapsto +  \\
*\ \mathit{isLocation} \ loc1 \mapsto +  \\
*\ \mathit{isLocation} \ loc2 \mapsto +  \\
*\ \mathit{isAt}\ v\ loc1 \mapsto - \\
*\ \mathit{isAt}\ v\ loc2 \mapsto +
\end{Bmatrix}$

\end{tabular}

Imagine we instantiated the $\mathit{move}$ action with $(\mathit{car}\ \mathit{museum}\ \mathit{museum})$.
In the PCP logic,
this instantiation  will produce an inconsistent state:

\noindent
\begin{tabular}{c}

$
\begin{Bmatrix}
\mathit{isVehicle\ car} \mapsto +  \\
*\ \mathit{isLocation \ museum} \mapsto +  \\
*\ \mathit{isLocation \ museum} \mapsto +  \\
*\ \mathit{isAt\ car\ museum} \mapsto +
\end{Bmatrix}

\strans

\begin{Bmatrix}
\mathit{isVehicle\ car} \mapsto +  \\
*\ \mathit{isLocation \ museum} \mapsto +  \\
*\ \mathit{isLocation \ museum} \mapsto +  \\
*\ \mathit{isAt\ car\ museum} \mapsto - \\
*\ \mathit{isAt\ car\ museum} \mapsto +
\end{Bmatrix}
$

\end{tabular}

\hfill

In PDDL the effect will simply be executed.
The result of this action
will depend on the order in which the effect formulas are executed.
And, since PDDL specifications~\cite{fox2003pddl2} do not specify
any particular ordering on effect formulas,
 planners have to
make this decision themselves.
So, some planners come to the conclusion that the car is at the museum, and some -- that it is not.

In the PCP logic, this plan will simply not be type-checked and the user will receive a due type checking error.

\textbf{Example 2. Loss in Translation} \label{loss}
In our early experiments, we encountered a problem that many good plans are not type-checked  when they are translated verbatim to the PCP logic.
The reason for this is the loss of information between the ``precondition'' and the ``effect'' in the PDDL formulation.
We use the following example to explain the problem.

Consider  the $ \mathit{pickup\_from\_stack}$ action
 from the BlockWorld domain definition:
\hfill\break
\begin{tabular}{cc}

$
[] ;

\begin{Bmatrix}
\mathit{on} \ x_{1} \ {x_{2}} \mapsto + \\
*\ \mathit{clear} \ x_{1} \mapsto +    \\
*\ \mathit{handEmpty} \mapsto +    \\
\end{Bmatrix}

\strans

\begin{Bmatrix}
\mathit{on} \ x_{1} \ {x_{2}} \mapsto - \\
*\ \mathit{handEmpty} \mapsto - \\
*\ \mathit{holding} \ x_{1} \mapsto + \\
*\ \mathit{clear} \ x_{2} \mapsto +
\end{Bmatrix}$

\end{tabular}

\hfill\break

\noindent Notice that $\mathit{clear} \ x_{1} \mapsto +$ is not mentioned in the effect list, because this fact is unaffected by the action. But if we treat this as a state,
rather than effect, the information about $\mathit{clear} \ x_{1} \mapsto +$ will simply be lost.
In the PCP logic, the frame rule can not be used to recover this information, as this formula already occurs in the precondition.
As a result, some PDDL plans will fail to type check in the PCP logic.
To fix this problem, we
add all such formula maps explicitly to the postconditions:

\begin{tabular}{cc}

$[] ;

\begin{Bmatrix}
\mathit{on} \ x_{1} \ {x_{2}} \mapsto + \\
*\ \mathit{clear} \ x_{1} \mapsto +    \\
*\ \mathit{handEmpty} \mapsto +    \\
\end{Bmatrix}

\strans

\begin{Bmatrix}
\mathit{on} \ x_{1} \ {x_{2}} \mapsto - \\
*\ \mathit{handEmpty} \mapsto - \\
*\ \mathit{holding} \ x_{1} \mapsto + \\
*\ \mathit{clear} \ x_{2} \mapsto + \\
*\ \highlightg{\mathit{clear} \ x_{1} \mapsto +}
\end{Bmatrix} $

\end{tabular}

The plan verifier we implement does this transformation automatically.

\section{Conclusions, Related and Future Work}\label{sec:RW}

We have presented the PCP logic,  a novel resource logic for verification of AI plans, and proven its soundness
relative to the possible world semantics of PDDL. We have shown the benefits that resource semantics
and the Curry-Howard correspondence bring to this framework. In particular, the former makes it easier
to formalise state consistency and other constraints within the logic,
and the latter enables direct deployment of verified plans as functions.
We also presented an Agda library in which the soundness result is proven, and
which simultaneously serves as a generic module for verifying plans produced by AI planning.
To further strengthen the practical significance of these results, we implemented
scripts for automated parsing of PDDL plans, and for automated generation of proofs of their soundness in the PCP logic.
The ultimate proof- (and type-) checking of these is delegated to the Agda library.
We evaluated this implementation on several famous PDDL benchmarks.

\textbf{Our Earlier Work on Proof Carrying Plans.}
Compared to our earlier attempts to define a ``proof-carrying plans'' approach~\cite{schwaab2019proof},
this new attempt is stronger both theoretically and practically.
The new PCP logic takes inspiration from resource logics
as a consequence provides a more
natural way to perform local reasoning. This, in its turn, helps to verify not just the plans,
but also consistency of domains and states. In previous work the
consistency assumption was needed to be stated as an axiom in order to prove soundness
of the system, and was not incorporated into checking of individual plans.
The PCP Logic
embeds consistency directly into the system through its rules.
This has two advantages. The first is that it is impossible to derive proofs
that contain inconsistent states and the second is that type errors for inconsistency
will show exactly where and why there is an inconsistent state. The PCP logic also
enables extensions to first-order logic, introduction of richer verification
constraints, and opens the possibilities for extensions to concurrent logics.
Though the latter extension is left as future work.

From the practical point of view, the earlier work contained no automation presented here.
Also, it did not include reasoning with constraints, or any experiments with using
the dependent types during the plan execution.

\textbf{Origins of the Frame Rule.}
The ``frame problem'' that inspired the frame rule of Separation logic
actually
has origins in AI~\cite{hayes1981frame,dennett2006cognitive}.
Initially, the problem referred to the difficulty in local reasoning about problems in a complex
world. In AI planning specifically, this problem consisted of keeping track
of the consequences of applying an action on a world. Intuitively a person would
understand picking up some block $a$ that is on the table would have no effect
on some other block $b$ that is on the table. The frame problem deals with the
way to represent this intuition formally.

One way to deal with the frame problem is to declare ``frame axioms'' for every action
explicitly. This is an inefficient way to deal with this
problem as defining these frame axioms becomes infeasible the larger the system
gets~\cite{dennett2006cognitive}. Since most actions in AI planning only make small
local changes to the world, a more general representation would be more suitable.
STRIPS deals with this problem by introducing an assumption that every formula
in a world that is not mentioned in the effect list of an action remains the same
after execution of the action. This is known as the ``STRIPS assumption'' and
it is an assumption that PDDL also uses.

The logic of Bunched Implications \cite{o2001local,ishtiaq2001bi} and Separation Logic~\cite{o2007resources}
took inspiration from this older notion of the frame problem, and introduced more abstract formalism,
which is now known as a ``frame rule'', into the resource logics~\cite{pym2019resource}.
This family of logics has brought many theoretical and practical advances to modelling of complex systems, and is behind many \emph{lightweight verification}
projects~\cite{calcagno2015moving}.

In this paper, we have shown how the original frame problem from AI maps back to the more
abstract ideas of resource logics.
We see this as one of the paper's contributions.

\textbf{Curry-Howard Approaches to Separation Logic and Other Resource Logics.}
The PCP Logic introduces a Curry-Howard approach to AI planning inspired by resource
logic. This is in part inspired by existing applications of the  Curry-Howard approach in the field. Both Hoare logic and Separation logic
have been given a Curry-Howard interpretation: \cite{nanevskiseparation,nanevski2006polymorphism}.
Several papers explore the computational and practical benefits of it.
For example, Polikarpova and Sergey \cite{polikarpova2019structuring}
took a Curry-Howard approach to Separation logic to improve program synthesis seen as a proof
search problem. In a similar way to our specifications, they define a
synthesis goal $\Gamma \vdash P \strans Q$, which is solved by a program $c$ if
the assertion $\Gamma \vdash  P \strans Q | c$ can be derived in their system.

In this paper, we also make an attempt to make a case for computational and practical uses of Curry-Howard interpretation of the
newly introduced PCP logic.

\textbf{AI Planning and Linear Logic.}
There is a long history of modelling AI planning in Linear logic, that dates back to the 90s~\cite{jacopin1993classical}, and was investigated in detail in the 2000s, see e.g.
\cite{chrpa2007encoding,steedman2002plans}.
In fact, AI planning is used as one of the iconic use-cases of Linear logic~\cite{polakow2001ordered}.
The main idea behind using Linear logic for AI planning is
treating action descriptions  as linear implications:
$$\alpha : \forall x. P \multimap Q,$$
where  $P$ and $Q$ are given by tensor products of atoms:
$R_1(t_1) \otimes \ldots \otimes R_n(t_n)$.
We could incorporate information about polarities inside the predicates, as follows:
$R_1(t_1,z_1) \otimes \ldots \otimes R_n(t_n,z_n)$.
Then, the linear implication and the tensor products model the resource semantics of PDDL rather elegantly.

The computational (Curry-Howard) interpretation of AI plans was not the focus of study in the above mentioned approaches,
yet it plays a crucial r\^{o}le in the PCP logic, from design all the way to implementation, verification and proof extraction (see Section~\ref{sec:implementation}).

\textbf{AI Planning and (Linear) Logic Programming.}
The above syntax also resembles linear logic programming Lolli, introduced by Miller et al~\cite{HodasM94}.
Lolli was applied in speech planning in~\cite{DixonST09}.

Our previous work~\cite{schwaab2019proof} in fact takes inspiration from Curry-Howard interpretation of Prolog\cite{0001KSP16,FuK17}.
In our previous work and in general, logic programming does not work well with PDDL negation. In PDDL, we have to work with essentially three-valued logic:
an atom may be declared to be absent or present in a world. But if neither is declared, we assume a ``not known'' or ``either'' situation.
Logic programming usually uses the approach of ``negation-as-failure'' that does not agree with this three-valued semantics.
A solution is to introduce polarities as terms, as shown in the example above. This merits further investigation.

\textbf{Curry-Howard view on Linear Logic.}
Curry-Howard semantics of Linear logic also attracted attention of logicians first in the 90s~\cite{AlbrechtCJ97}, and then in the 2000s
in connection with research into Linear Logical Frameworks~\cite{Schack-NielsenS08,CervesatoP02}.

We conjecture that many results obtained in this paper could be replicated in one of these systems.
We plan to investigate this approach in comparison with the PCP logic in the future.
Generally speaking, the PCP logic can be seen as a domain specific language for AI planning. It is simpler and
less expressive than Linear logic but makes up for it in simplicity and close correspondence
to PDDL syntax. Transformations between PDDL domain and problem descriptions to the PCP logic
are straightforward since the syntax is so similar.
This enables us to automate the generation of Agda proofs from PDDL plans.
Notably, we have typing rules for functions
that are given directly by PDDL plans.
Thus, we verify outputs of PDDL planners as given.
This close correspondence to the plans would be impossible in either of the above Curry-Howard versions of Linear logic, where proof terms tend to be much more complex.
Pros and cons of domain specific versus general approaches to verification of AI plans deserves further investigation.

We hope that the DSL nature of the PCP logic will pave the way for its wider adoption as a \emph{practical verification tool}  for the AI planning
community. This is something that previously proposed Linear logic approaches to AI planning did not achieve.

\textbf{Modelling looping behaviour and non-termination in AI planning.}
The design of this Agda prototype has revealed several limitations in state-of-the-art implementations of planning languages: e.g. their
reliance on the closed word assumption and formulae grounding and the absence of functions. We see the potential of our method to overcome many of these limitations thanks to our general dependently-typed set-up, in which the use of functions, higher-order features, constraints and effect handling will be much more natural than in the current implementations.

\textbf{Other Future Work.}
One limitation of the PCP logic is that it only works with a subset of the domains
that can be expressed in PDDL. To incorporate more of the PDDL syntax
we want to extend the system to reason about temporal (as well as concurrent) planning.
We believe that this extension can be naturally expressed
in our system due to related extensions in the resource logics.

From the theoretical point of view,
we hope to achieve  a deeper understanding of the relation of the new PCP logic to the categorical and coalgebraic semantics of
other resource logics~\cite{pym2019resource}.

We plan to improve the performance of our system, to speed up type checking, and make
Agda proof generation more reliable and practical.
The former can be improved through the creation of a frame minimising algorithm. The latter
can be facilitated
by producing partial Agda proofs when the full proof generation is too hard.

Interactive facilities of our tool also deserve future attention. Generally, Agda allows \emph{holes} to be left in
a proof which a user can use to interactively inspect the subgoal of the proof. In the future
we plan to update our proof generator to generate incomplete proofs so a user
can inspect the proof goals that cannot be solved.

Another possibility is to further explore the dependently-typed aspects of our system 
as described in Section~\ref{sec:dependent}. This can include
extensions such as higher-order functions and universal formulae.


\begin{acks}
We thank Simon Docherty, James McKinna and David Pym, for inspiring discussions, constructive criticisms and pointers to related work.

We thank EPSRC DTA PhD Scheme  for funding the first author.

The second author acknowledges support of the UK Nationaly Cyber Security Center grant \emph{SecCon-NN: Neural Networks with Security Contracts - towards lightweight, modular security for neural networks}
and the UK Research Institute in Verified Trustworthy Software Systems (VETSS)-funded research project \emph{CONVENER: Continuous Verification of Neural Networks}.
\end{acks}

\bibliographystyle{ACM-Reference-Format}
\bibliography{mybib}

\pagebreak

\end{document}